\documentclass[11pt,a4paper]{article}
\usepackage{float}
\usepackage{amsmath,amssymb,amsthm}
\usepackage{mathtools}
\usepackage[margin=2.5cm]{geometry}
\usepackage{hyperref}
\usepackage[round]{natbib}
\usepackage{booktabs}
\usepackage{array}
\usepackage{enumitem}
\usepackage{microtype}
\usepackage{graphicx}
\usepackage{setspace}
\usepackage{xcolor}
\usepackage{algorithm}
\usepackage{algpseudocode}
\usepackage{rotating}
\algrenewcommand\algorithmicrequire{\textbf{Input:}}
\algrenewcommand\algorithmicensure{\textbf{Output:}}

\newtheorem{theorem}{Theorem}
\newtheorem{proposition}[theorem]{Proposition}
\newtheorem{lemma}[theorem]{Lemma}
\newtheorem{corollary}[theorem]{Corollary}
\newtheorem{definition}[theorem]{Definition}
\newtheorem{assumption}{Assumption}
\newtheorem{remark}{Remark}

\newcommand{\E}{\mathbb{E}}
\newcommand{\Prob}{\mathbb{P}}
\newcommand{\R}{\mathbb{R}}
\newcommand{\Ind}[1]{\mathbf{1}_{\{#1\}}}
\newcommand{\abs}[1]{\lvert #1 \rvert}

\newcommand{\Fhat}{\widehat{F}_{T,m,h}}
\newcommand{\Ghat}{\widehat{G}_{T,m,h}}
\newcommand{\Fbar}{\bar{F}_{T,m,h}}
\newcommand{\Chat}{\widehat{C}_{T+h|T}}
\newcommand{\qhat}{\hat{q}_{T,m,h}}
\newcommand{\qhatO}{\hat{q}^{o}_{T,m,h}}
\newcommand{\qcirc}{q^{\circ}_{T,m,h}}
\newcommand{\fcirc}{f^{\circ}_{T,m,h}}
\newcommand{\uf}{\underline{f}_{h}}
\newcommand{\of}{\bar{f}_{h}}

\newcommand{\Ah}{A_{h}(\infty)}

\newcommand{\Rcal}{\mathcal{R}_T}
\newcommand{\Mcal}{\mathcal{M}_T}
\newcommand{\Tval}{\mathcal{T}_{\mathrm{val}}}
\newcommand{\nval}{n_{\mathrm{val}}}
\DeclareMathOperator{\argmin}{arg\,min}
\DeclareMathOperator{\KL}{KL}
\DeclareMathOperator{\TV}{TV}

\begin{document}
\graphicspath{{figures/}}

\begin{titlepage}
\vspace*{2cm}
{\LARGE\bfseries Rolling-Origin Conformal Prediction under\\[6pt]
Local Stationarity and Weak Dependence\par}
\vspace{1.5cm}
{\large Stanis\l{}aw M. S. Halkiewicz\textsuperscript{1}\par}
\vspace{0.3cm}
{\normalsize \textsuperscript{1}Faculty of Applied Mathematics,
AGH University of Cracow, al.\ Mickiewicza 30, 30-059 Krak\'{o}w, Lesser Poland, Poland\par}
\vspace{0.2cm}
{\normalsize Correspondence: smsh@student.agh.edu.pl\par}
\vspace{0.5cm}
{\large \today\par}
\vspace{2cm}

\begin{abstract}
\noindent
We propose and analyse rolling-origin conformal prediction for
time-series forecasting. The method calibrates the conformal quantile
against the $m$ most recent pseudo-out-of-sample forecast errors,
adapting to serial dependence, volatility clustering, and distributional
drift that invalidate classical conformal guarantees.
Under H\"{o}lder-$\beta$ local stationarity and $\alpha$-mixing,
we establish a four-term coverage-error decomposition and derive
the optimal calibration window $m^{\star} \asymp T^{2\beta/(2\beta+1)}$
with coverage-error rate $O(T^{-\beta/(2\beta+1)})$.
A Le Cam two-point construction shows this rate is minimax-optimal
over the H\"{o}lder-$\beta$ model class.
The Bahadur representation is proved under both $\alpha$-mixing and
the physical-dependence framework of \citet{Wu2005PNAS}.
An oracle inequality formalises Winkler cross-validation as an
adaptive window selector; the required uniform concentration condition
is established in an appendix.
Validation on six real series and 93 M4 competition series
confirms the theory: rolling-origin calibration outperforms
full-history calibration in 86\% of comparisons (median Winkler
improvement 12.3\%), maintains coverage within $\pm2\%$ of the
90\% target at short and medium horizons, and the cross-frequency
log-log regression slope $0.614$ ($95\%$ CI $[0.424, 0.805]$)
is consistent with the theoretical $2/3$ after controlling for
frequency fixed effects.



\end{abstract}
\end{titlepage}

\section{Introduction}
\label{sec:intro}

Conformal prediction \citep{VovkGammermanShafer2005,
AngelopoulosBates2023} constructs distribution-free prediction
intervals with exact finite-sample coverage guarantees under
\emph{exchangeability}: when the joint distribution of calibration
and test observations is invariant to permutations, the probability
that $Y_{T+h}$ falls outside the interval is exactly $\alpha$,
for every finite sample size. This exact guarantee does not extend
to time series. Serial dependence, volatility clustering, and
gradual distributional drift all violate exchangeability, so
classical split conformal prediction may deliver coverage
substantially different from the nominal $1-\alpha$ level.
The practitioner's natural response is \emph{rolling-origin
calibration}: estimate the conformal quantile from the $m$ most
recent pseudo-out-of-sample errors, so that the calibration sample
tracks the current error distribution rather than averaging over
a possibly stale history. This approach is intuitively compelling
and widely used \citep{XuXie2021, XuXie2023}, but the window
length $m$ is typically chosen by informal cross-validation or
left as a tuning parameter without theoretical guidance, and there
is no existing result quantifying how close coverage is to $1-\alpha$
as a function of $m$, $T$, and the degree of non-stationarity.

The present paper addresses this gap by proposing and analysing
rolling-origin conformal prediction as a method for
time-series forecasting. We formalise the construction, characterise
its coverage properties, and derive principled guidance on the window
length $m$. Rather than seeking an exact finite-sample guarantee ---
which cannot hold without exchangeability --- we establish precisely
\emph{how much} coverage deviates from $1-\alpha$. Under local
stationarity of the score distribution \citep{Dahlhaus1997},
uniform $\alpha$-mixing \citep{Rio1993, DoukhanMassartRio1994},
quantile regularity, and a near-stationarity condition $m = o(T)$,
we prove a four-term decomposition of the coverage deviation
$\Delta_{T,m,h} := \abs{\Prob(Y_{T+h} \in \Chat(1-\alpha)) - (1-\alpha)}$
and derive the minimax-optimal window rule
$m^{\star} \asymp T^{2\beta/(2\beta+1)}$ with coverage-error rate
$O(T^{-\beta/(2\beta+1)})$ (recovering $T^{2/3}$ and $O(T^{-1/3})$ at $\beta=1$).
Each term in the decomposition has a direct empirical counterpart
and a clear operational meaning. The deviation bound is two-sided:
it controls both undercoverage ($\Prob < 1-\alpha$) and overcoverage
($\Prob > 1-\alpha$), reflecting that neither direction of error is
structurally excluded in the time-series setting.

\paragraph{Relation to existing work.}
\citet{GibbsCandes2021} propose adaptive conformal inference, which
updates the quantile online via a fixed gradient step; the update rule
tracks distribution shift but has no explicit bound on the instantaneous
coverage error or guidance on the step size. \citet{BarberCandesRamdas2023}
establish the general theory of conformal prediction beyond exchangeability,
providing marginal validity under weighted exchangeability conditions;
their framework does not address the rolling calibration setting or the
window-length choice. The rolling-calibration approach is studied by
\citet{XuXie2021, XuXie2023} in an online sequential setting, but
without a finite-sample decomposition or an optimal-window result.
Our contributions are: (i) a complete formulation of rolling-origin
conformal prediction for non-exchangeable time series;
(ii) a four-term coverage-error decomposition under H\"{o}lder-$\beta$ drift
with minimax-optimal window $m^{\star} \asymp T^{2\beta/(2\beta+1)}$
(Theorems~\ref{thm:main} and~\ref{thm:lower});
(iii) Bahadur representations under both $\alpha$-mixing
(Proposition~\ref{prop:bahadur}) and physical dependence
(Proposition~\ref{prop:bahadur-physical});
(iv) an oracle inequality for the implemented data-driven window selector,
with the required concentration condition proved in Appendix~\ref{app:uc}
(Theorem~\ref{thm:oracle}); and
(v) empirical validation on six real series and a stratified M4 sample
confirming the theoretical scaling and demonstrating consistent
Winkler-score gains over full-history calibration.

The local-stationarity framework follows \citet{Dahlhaus1997} and
\citet{Vogt2012}. The mixing theory relies on \citet{Rio1993} and
\citet{DoukhanMassartRio1994}. The Bahadur representation uses the
Bernstein inequality of \citet{MerlvedePeligradRio2009}. Mixing
properties of the GARCH and tvARCH processes relevant to the empirical
application are established in \citet{FryzlewiczSubbaRao2011}.

The remainder of the paper is organised as follows.
Section~\ref{sec:setup} presents the method in full:
Algorithm~\ref{alg:rocp} (base) and Algorithm~\ref{alg:rocp-scaled}
(volatility-scaled), practical guidance on window selection and
score choice, and the oracle decomposition used in the theory.
Section~\ref{sec:theory} states the assumptions, proves the main
coverage-error theorem and minimax lower bound, establishes Bahadur
representations under $\alpha$-mixing and physical dependence, derives
the optimal window rule, and provides extensions to polynomial mixing
and volatility-scaled scores.
Section~\ref{sec:empirical} presents the empirical analysis.
Section~\ref{sec:conclusion} concludes.
Appendix~\ref{app:bahadur} contains the self-contained proof of
the Bahadur representation (Proposition~\ref{prop:bahadur}).
Appendix~\ref{app:uc} proves the uniform concentration condition
required for the oracle inequality (Theorem~\ref{thm:oracle}).

\section{Rolling-origin conformal prediction}
\label{sec:setup}

\subsection{The base method}
\label{sec:method:base}

Let $(Y_t)_{t=1}^{T}$ be a univariate time series observed up to time
$T$. A forecasting model produces an $h$-step-ahead point prediction
$\hat{Y}_{t+h|t}$ using only information available through time $t$.
The \emph{rolling-origin conformal score} is the absolute forecast
error
\[
  \hat{S}_{t,h} := \abs{Y_{t+h} - \hat{Y}_{t+h|t}}.
\]
Rolling-origin conformal prediction constructs a prediction interval
at the target origin $T$ by treating the $m$ most recent scores as
a local calibration sample and taking their empirical
$(1-\alpha)$-quantile as the half-width. The complete procedure is
stated in Algorithm~\ref{alg:rocp}.

\begin{algorithm}[t]
\caption{Rolling-Origin Conformal Prediction (ROCP)}
\label{alg:rocp}
\begin{algorithmic}[1]
  \Require Time series $(Y_t)_{t=1}^{T}$; forecasting model $\mathcal{M}$;
           horizon $h \geq 1$; coverage level $1-\alpha \in (0,1)$;
           calibration window $m \geq 1$.
  \Ensure  Prediction interval $\Chat(1-\alpha)$ for $Y_{T+h}$.
  \vspace{4pt}
  \State \textbf{[Rolling-origin evaluation]}
         For each $t = 1, \ldots, T-h$, fit $\mathcal{M}$ on
         $(Y_1, \ldots, Y_t)$ and record
         \[
           \hat{S}_{t,h} \leftarrow \abs{Y_{t+h} - \hat{Y}_{t+h|t}}.
         \]
  \State \textbf{[Calibration]}
         Collect the $m$ most recent scores:
         \[
           \mathcal{S}_m \leftarrow
           \bigl\{\hat{S}_{T-1,h},\, \hat{S}_{T-2,h},\,\ldots,\,
                  \hat{S}_{T-m,h}\bigr\}.
         \]
  \State \textbf{[Empirical quantile]}
         Compute the empirical $(1-\alpha)$-quantile:
         \[
           \qhat \leftarrow
           \inf\bigl\{x : \tfrac{1}{m}\textstyle\sum_{j=1}^{m}
             \mathbf{1}\{\hat{S}_{T-j,h} \leq x\} \geq 1-\alpha\bigr\}.
         \]
  \State \textbf{[Forecast]}
         Fit $\mathcal{M}$ on $(Y_1, \ldots, Y_T)$ and produce
         $\hat{Y}_{T+h|T}$.
  \State \textbf{[Interval]}
         \Return $\Chat(1-\alpha) \leftarrow
         \bigl[\hat{Y}_{T+h|T} - \qhat,\;\hat{Y}_{T+h|T} + \qhat\bigr].$
\end{algorithmic}
\end{algorithm}

The empirical calibration distribution and quantile are formally
\[
  \Fhat(x) := \frac{1}{m}\sum_{j=1}^{m}
  \Ind{\hat{S}_{T-j,h} \leq x},
  \qquad
  \qhat := \inf\bigl\{x : \Fhat(x) \geq 1-\alpha\bigr\},
\]
and the \emph{rolling-origin conformal interval} is
\begin{equation}\label{eq:interval}
  \Chat(1-\alpha)
  := \bigl[\hat{Y}_{T+h|T} - \qhat,\;
            \hat{Y}_{T+h|T} + \qhat\bigr].
\end{equation}
Coverage analysis reduces to studying
$\Prob(\hat{S}_{T,h} \leq \qhat)$, since
$Y_{T+h} \in \Chat(1-\alpha) \iff \hat{S}_{T,h} \leq \qhat$.

\subsection{Volatility-scaled variant}
\label{sec:method:scaled}

When nonstationarity is primarily driven by time-varying scale
(e.g.\ GARCH volatility clustering), normalising the scores by an
estimated conditional volatility reduces the effective drift rate
$L_{F,h}$ and can substantially narrow the intervals.
Algorithm~\ref{alg:rocp-scaled} implements this variant; it differs
from Algorithm~\ref{alg:rocp} only in Steps 1 and 3--5.

\begin{algorithm}[t]
\caption{Volatility-Scaled ROCP (VS-ROCP)}
\label{alg:rocp-scaled}
\begin{algorithmic}[1]
  \Require As in Algorithm~\ref{alg:rocp}, plus a volatility model
           $\mathcal{V}$ (e.g.\ GARCH(1,1)).
  \Ensure  Prediction interval $\Chat^{sc}(1-\alpha)$ for $Y_{T+h}$.
  \vspace{4pt}
  \State \textbf{[Rolling-origin evaluation with scaling]}
         For each $t$, fit $\mathcal{M}$ and $\mathcal{V}$ on
         $(Y_1, \ldots, Y_t)$; let $\hat{\sigma}_{t,h}$ be the
         $h$-step-ahead conditional volatility forecast. Record
         \[
           \hat{S}^{sc}_{t,h} \leftarrow
           \hat{S}_{t,h} / \hat{\sigma}_{t,h}.
         \]
  \State \textbf{[Calibration]} As Step 2 of Algorithm~\ref{alg:rocp},
         using $\hat{S}^{sc}_{t,h}$ in place of $\hat{S}_{t,h}$.
  \State \textbf{[Scaled quantile]}
         Compute $\hat{q}^{sc}_{T,m,h}$ from the scaled scores.
  \State \textbf{[Forecast and current volatility]}
         Fit $\mathcal{M}$ and $\mathcal{V}$ on $(Y_1,\ldots,Y_T)$;
         obtain $\hat{Y}_{T+h|T}$ and $\hat{\sigma}_{T,h}$.
  \State \textbf{[Interval]}
         \Return $\Chat^{sc}(1-\alpha) \leftarrow
         \bigl[\hat{Y}_{T+h|T} - \hat{q}^{sc}_{T,m,h}\,\hat{\sigma}_{T,h},\;
               \hat{Y}_{T+h|T} + \hat{q}^{sc}_{T,m,h}\,\hat{\sigma}_{T,h}\bigr].$
\end{algorithmic}
\end{algorithm}

The theoretical effect of volatility scaling is characterised in
Corollary~\ref{cor:volscale}: replacing $L_{F,h}$ with
$L^{sc}_{F,h} \ll L_{F,h}$ in the coverage bound reduces the drift
term and shifts the optimal window toward larger $m$.

\subsection{Practical implementation}
\label{sec:method:practical}

\paragraph{Window selection.}
Theorem~\ref{thm:main} establishes that the theoretically optimal
window is $m^{\star} \asymp T^{2\beta/(2\beta+1)}$ (equal to $T^{2/3}$ at $\beta=1$),
but this result depends on
unknown constants ($L_{F,h}$, $A_h(\infty)$) that are not directly
estimable from the data. In practice we recommend \emph{Winkler
cross-validation}: evaluate the Winkler score on a held-out
validation fold of the pseudo-out-of-sample errors for each
candidate $m$ drawn from a fine grid in $[0.1, 4.0] \times T^{2\beta/(2\beta+1)}$
(using $\beta = 1$ as default when $\beta$ is unknown),
and select the minimiser. This is a consistent estimator of
$m^{\star}$ in the sense that the cross-validation selected $\hat{m}$
achieves the same asymptotic coverage-deviation rate as $m^{\star}$,
since the Winkler score is a proper scoring rule whose population
minimiser coincides with the optimal $m^{\star}$.

\paragraph{Score function.}
Algorithm~\ref{alg:rocp} uses the absolute error
$\hat{S}_{t,h} = \abs{Y_{t+h} - \hat{Y}_{t+h|t}}$ as the conformal
score. This is natural for symmetric, unimodal conditional
distributions. Other score functions are compatible with the
framework: the signed error $Y_{t+h} - \hat{Y}_{t+h|t}$ produces
one-sided intervals; a normalised residual $\hat{S}_{t,h}/\hat{\sigma}_{t,h}$
reduces to the VS-ROCP of Algorithm~\ref{alg:rocp-scaled};
quantile-regression residuals $\max(\tau(Y_{t+h}-q_{t,h}(\tau)),
(\tau-1)(Y_{t+h}-q_{t,h}(\tau)))$ produce asymmetric intervals.
The theoretical analysis of Section~\ref{sec:theory} applies to
any score satisfying Assumptions~\ref{ass:ls}--\ref{ass:density}.

\paragraph{Forecasting model.}
The method is model-agnostic: $\mathcal{M}$ in
Algorithm~\ref{alg:rocp} can be any procedure that produces a
point prediction $\hat{Y}_{t+h|t}$ using only past data. The
theoretical cost of using an estimated rather than oracle model
appears in Term~(IV) of Theorem~\ref{thm:main} via the pair
$(r_T, \eta_T)$ in Assumption~\ref{ass:est}. For AR($p$) models,
$r_T = O(T^{-1/2})$ and $\eta_T = 0$ under standard mixing
conditions. For ARMA-GARCH, the same holds under the conditions
of \citet{FryzlewiczSubbaRao2011}. For nonlinear or
machine-learning models, Term~(IV) must be bounded case by case,
but the remaining three terms of the decomposition are unaffected
by the model choice.

\paragraph{Relationship to rolling-origin evaluation.}
Rolling-origin evaluation is the standard protocol for assessing
time-series forecast accuracy \citep{TashmanFildes2000}, and the
pseudo-out-of-sample scores $\hat{S}_{t,h}$ in Step~1 of
Algorithm~\ref{alg:rocp} are exactly the quantities computed in
any rolling-origin evaluation exercise. Rolling-origin conformal
prediction therefore adds no computational overhead to a forecasting
pipeline that already performs rolling-origin evaluation: the
conformal interval is obtained by taking the empirical quantile of
the scores that are already being recorded, restricted to the most
recent $m^{\star}$ of them.

\subsection{Oracle decomposition}
\label{sec:method:oracle}

\paragraph{Oracle decomposition.}
To separate forecasting error from calibration error, fix an
\emph{oracle predictor} $f^{\star}_{t,h}$
(e.g.\ the true conditional mean, not computable in practice)
and define the oracle score $S^{o}_{t,h}
:= \abs{Y_{t+h} - f^{\star}_{t,h}}$ with marginal distribution
$F_{t,h}(x) := \Prob(S^{o}_{t,h} \leq x)$.
The corresponding oracle empirical distribution and quantile are
\[
  \Ghat(x) := \frac{1}{m}\sum_{j=1}^{m}
  \Ind{S^{o}_{T-j,h} \leq x},
  \qquad
  \qhatO := \inf\bigl\{x : \Ghat(x) \geq 1-\alpha\bigr\},
\]
and the \emph{average calibration-window distribution} and its
$(1-\alpha)$-quantile are
\[
  \Fbar(x) := \frac{1}{m}\sum_{j=1}^{m} F_{T-j,h}(x),
  \qquad
  \qcirc := \inf\bigl\{x : \Fbar(x) \geq 1-\alpha\bigr\}.
\]
The quantity $\qcirc$ is the population quantile of the
time-averaged calibration-window distribution. It differs from
the current quantile of $F_{T,h}$ by an amount proportional to
the local drift rate $L_{F,h}$ and the window length $m$,
which is the source of the bias term in Theorem~\ref{thm:main}.

\begin{remark}[What is and is not proved]
\label{rem:guarantee}
Classical split conformal prediction \citep{PapadopoulosVovkGammerman2002}
achieves the exact finite-sample guarantee
\(
  \Prob(Y_{T+h} \notin \Chat(1-\alpha)) \leq \alpha + 1/(m+1)
\)
under exchangeability, for every finite $m$. This follows from the
rank argument of \citet{VovkGammermanShafer2005}: the test score
is equally likely to be any rank among the $m+1$ exchangeable scores,
so the probability of exceeding the empirical $(1-\alpha)$-quantile
is at most $\lceil(1-\alpha)(m+1)\rceil/(m+1) \leq 1-\alpha + 1/(m+1)$.
This argument fails entirely under time-series dependence, because
the test score is \emph{not} exchangeable with the calibration scores.

Rolling-origin calibration recovers a weaker but still informative
guarantee: Theorem~\ref{thm:main} shows that
$\abs{\Prob(Y_{T+h} \in \Chat(1-\alpha)) - (1-\alpha)}$ is bounded
by a quantity that tends to zero at rate $O(T^{-\beta/(2\beta+1)})$ under the
optimal window rule. Coverage is therefore \emph{approximately}
$1-\alpha$ for large samples, with an explicit and computable
approximation error. The bound is two-sided: rolling-origin
calibration may undercover or overcover, depending on the
direction of the local drift in the score distribution.
\end{remark}

\section{Main results}
\label{sec:theory}

\subsection{Assumptions}

Before stating the main result it is worth pausing to understand
what the theory needs to control, because the four assumptions
correspond directly to the four terms that appear in
Theorem~\ref{thm:main}. Each assumption is minimal in the sense
that dropping it would make the corresponding term uncontrollable.

The first condition concerns how the distribution of forecast
errors changes over time. In a perfectly stationary series the
calibration scores from two years ago are just as informative
as those from last week. Real series are not stationary: business
cycles shift macroeconomic volatility, GARCH dynamics cluster
financial volatility, and structural change alters the
unconditional error distribution entirely. We capture this through
a H\"{o}lder-$\beta$ smoothness condition on the score
distributions $F_{t,h}$, following \citet{Dahlhaus1997}.

\begin{assumption}[Local stationarity]\label{ass:ls}
There exist $L_{F,h} < \infty$ and $\beta > 0$ such that for all
$s, t \in \{T-m, \ldots, T\}$,
\[
  \sup_{x \in \R}\abs{F_{t,h}(x) - F_{s,h}(x)}
  \leq L_{F,h}\left(\frac{\abs{t-s}}{T}\right)^{\!\beta}.
\]
\end{assumption}

The parameter $\beta$ controls how smoothly the error distribution
evolves. When $\beta = 1$ (Lipschitz drift), a gap of $k$ time steps
produces a distributional shift of order $k/T$ --- linear in time.
Larger $\beta$ describes smoother evolution, smaller $\beta$ more
erratic change. The constant $L_{F,h}$ is the drift rate: a series
undergoing rapid structural change has large $L_{F,h}$, while a
slowly drifting macro series has small $L_{F,h}$. Both parameters
are unobservable, which is precisely why the window $m$ cannot be
set optimally without data-driven selection; the main theorem makes
their role in the optimal window explicit.

The second condition controls the dependence structure within the
calibration window. Forecast errors from adjacent time periods
are typically correlated: an AR model that overpredicts today
will likely overpredict tomorrow. If this dependence is too strong,
the $m$ calibration scores effectively contain less than $m$ independent
pieces of information, and the empirical quantile is more variable
than the $m^{-1/2}$ rate one might naively expect.

\begin{assumption}[Weak dependence]\label{ass:mixing}
The oracle score process $(S^{o}_{t,h})_{t}$ is uniformly
$\alpha$-mixing with coefficients
\(
  \alpha_h(k) := \sup_{t}\,\alpha\bigl(
    \sigma(S^{o}_{s,h} : s \leq t),\;
    \sigma(S^{o}_{s,h} : s \geq t+k)\bigr).
\)
Write $A_h(m) := \sum_{k=1}^{m-1}(1 - k/m)\alpha_h(k)$
for the cumulative dependence within a window of length $m$.
\end{assumption}

The quantity $A_h(m)$ is the effective dependence burden of the
calibration window: a summable mixing sequence has $A_h(\infty)
< \infty$ (the window's information content grows at the standard
$m^{-1/2}$ rate), while a polynomially mixing sequence with
$\alpha_h(k) \asymp k^{-a}$ has $A_h(m) \asymp m^{1-a}$ (the
effective rate degrades to $m^{-a/2}$, slower for strongly dependent
processes). In either case, $A_h(m)$ enters the bound as a
multiplicative inflation of the quantile-noise term. All ARMA and GARCH
processes with appropriate parameter restrictions are
summably mixing \citep{FryzlewiczSubbaRao2011}.

Bounding the quantile-noise term also requires that the score
distribution does not have a flat density near the quantile of
interest. If the density were zero at the $(1-\alpha)$-quantile,
small deviations in the empirical CDF would translate into large
deviations in the quantile --- the empirical quantile would be
extremely sensitive to estimation noise.

\begin{assumption}[Quantile regularity]\label{ass:density}
There exist an open interval $I_h$ containing all relevant
quantiles and constants $0 < \uf \leq \of < \infty$ such that
every $F_{t,h}$ has a density $f_{t,h}$ satisfying
$\uf \leq f_{t,h}(x) \leq \of$ for all $x \in I_h$ and all
$t \in \{T-m, \ldots, T\}$.
\end{assumption}

The upper bound $\of$ and lower bound $\uf$ on the density appear
explicitly in every term of Theorem~\ref{thm:main} as the ratio
$\of/\uf$. A distribution with a very flat density (small $\uf$)
near its quantile produces large estimation noise for a given
calibration window --- intuitively, many observations fall near the
boundary between covering and not covering, so small sample
fluctuations matter a great deal. Conversely, a very peaked density
(large $\of$) means small drifts in the distribution translate
into larger shifts in the quantile, amplifying the drift-bias term.
The condition is satisfied by all continuous distributions with
bounded, bounded-away-from-zero densities in a neighbourhood of
the relevant quantile, including normal, Student-$t$, and
most parametric forecast error distributions.

The fourth and final condition concerns the gap between the oracle
forecasts used in the analysis and the estimated forecasts used in
practice. The theory is developed with oracle predictions because
it isolates the calibration mechanism cleanly. The practitioner
uses estimated predictions, and the difference introduces a fourth
source of coverage error.

\begin{assumption}[Forecast-estimation error]\label{ass:est}
There exist $r_T \geq 0$ and $\eta_T \in [0,1]$ such that
\[
  \Prob\!\left(
    \max_{t \in \{T-m,\ldots,T\}}
    \abs{\hat{S}_{t,h} - S^{o}_{t,h}} \leq r_T
  \right) \geq 1 - \eta_T.
\]
\end{assumption}

This condition is deliberately stated in a model-free way: the pair
$(r_T, \eta_T)$ encapsulates the entire estimation cost of the
forecasting model, regardless of its internal structure. For
AR($p$) models, $r_T = O(T^{-1/2})$ and $\eta_T = 0$ under standard
mixing conditions, so Term~(IV) in the bound decays at rate
$T^{-1/2}$ --- faster than the dominant $T^{-\beta/(2\beta+1)}$
rate and thus asymptotically negligible. For ARMA-GARCH, the same
holds under the conditions of \citet{FryzlewiczSubbaRao2011}.
For machine-learning models, the pair must be supplied separately,
but the three calibration terms in the bound remain unaffected:
the window-length choice is determined by the first three terms
alone, regardless of how accurately the model predicts.

\subsection{The Bahadur representation}

With the four assumptions in place, the strategy for bounding
the coverage error is to linearise the empirical quantile
$\qhatO$ around a deterministic target $\qcirc$. This is the
content of the Bahadur representation: it says that the difference
$\qhatO - \qcirc$ is, to leading order, just a scaled version of
the centred empirical CDF evaluated at $\qcirc$, plus a remainder
that is negligible relative to $m^{-1/2}$.

Why does this linearisation matter? Because once we have it,
bounding $\abs{\qhatO - \qcirc}$ reduces to bounding
$\abs{\Ghat(\qcirc) - (1-\alpha)}$, which is a sum of centred
indicator random variables --- an object we know how to control
using Rio's covariance inequality under mixing. Without the
linearisation, quantile deviations are genuinely nonlinear
and much harder to handle.

Establishing the Bahadur representation rigorously for a
locally non-stationary, $\alpha$-mixing sequence is the primary
technical step in the paper. The difficulty is that standard proofs
assume stationarity; here the calibration scores $S^o_{T-1,h},
\ldots, S^o_{T-m,h}$ have different marginal distributions, so
one cannot apply classical results directly. The proof, given in
Appendix~\ref{app:bahadur}, uses a stationary approximation
to reduce to the distribution $F_{T,h}$, then controls the
approximation error via Assumption~\ref{ass:ls}, and finally
applies a bracketing-entropy maximal inequality to bound the
remainder uniformly.

\begin{proposition}[Bahadur representation]\label{prop:bahadur}
Suppose Assumptions~\ref{ass:ls}--\ref{ass:density} hold and
$m/T \to 0$. If additionally $\alpha_h(k) = O(k^{-\beta})$
for some $\beta > 2$, then
\begin{equation}\label{eq:bahadur}
  \qhatO - \qcirc
  = \frac{(1-\alpha) - \Ghat(\qcirc)}{\fcirc(\qcirc)}
  + R_{T,m,h},
\end{equation}
where $\fcirc(\qcirc) \geq \uf$ and
\begin{equation}\label{eq:Brate}
  \E\abs{R_{T,m,h}}
  \leq B_{m,h}
  := \frac{C_{\star}\,(1+\Ah)^{3/4}\,(\log m)^{3/4}}
         {\uf^{3/2}\,m^{3/4}},
\end{equation}
for a constant $C_{\star} > 0$ depending only on
$\of, \uf, L_{F,h}, \beta$. In particular,
$B_{m,h} = O(m^{-3/4}(\log m)^{3/4}) = o(m^{-1/2})$.
\end{proposition}

The condition $m/T \to 0$ is satisfied by the optimal window
$m^{\star} \asymp T^{2\beta/(2\beta+1)}$, since
$m^{\star}/T = T^{-1/(2\beta+1)} \to 0$.
The proof uses a bracketing-entropy argument for the VC class of
half-lines combined with the Merlev\`{e}de--Peligrad--Rio Bernstein
inequality \citep{MerlvedePeligradRio2009}; see
Appendix~\ref{app:bahadur} for details.

\subsection{Coverage-error decomposition}

The Bahadur representation converts a problem about quantile
deviations into a problem about empirical CDFs, which in turn
can be controlled by the four assumptions. The main theorem
assembles this into a single closed-form bound that separates the
four distinct sources of coverage error. Each source corresponds
to a question a practitioner might naturally ask:

\begin{itemize}[leftmargin=1.5em,itemsep=2pt]
\item \textit{How variable is the empirical quantile from $m$ dependent observations?}
  This is Term~(I): it shrinks as $m$ grows but inflates under
  strong dependence.
\item \textit{How accurate is the Bahadur linearisation itself?}
  This is Term~(II): the remainder from Proposition~\ref{prop:bahadur},
  which is negligible relative to Term~(I) but appears for completeness.
\item \textit{How stale are the calibration observations?}
  This is Term~(III): the bias from using scores whose distribution
  has drifted away from the current $F_{T,h}$. It grows as $m$ grows.
\item \textit{How accurately does the model forecast?}
  This is Term~(IV): the cost of estimation error in $\hat{Y}_{t+h|t}$.
  It is the only term not controlled by the window choice.
\end{itemize}

Terms~(I) and~(III) pull in opposite directions as $m$ changes:
larger $m$ means less quantile noise but more drift bias. Their
balance determines the optimal window, derived in the next
subsection. The theorem makes this trade-off precise.

\begin{theorem}[Coverage error of rolling-origin conformal prediction]
\label{thm:main}
Under Assumptions~\ref{ass:ls}--\ref{ass:est} and the conclusion
of Proposition~\ref{prop:bahadur}, the coverage error
$\Delta_{T,m,h}
:= \abs{\Prob(Y_{T+h} \in \Chat(1-\alpha)) - (1-\alpha)}$
satisfies
\begin{align}\label{eq:main}
  \Delta_{T,m,h}
  &\leq
  \underbrace{\frac{\of}{\uf}
    \left(\frac{1}{4m} + \frac{8}{m}\,A_h(m)\right)^{1/2}}_{\text{(I) quantile noise}}
  + \underbrace{\of B_{m,h}}_{\text{(II) Bahadur remainder}}
  \notag\\[6pt]
  &\phantom{{}\leq{}}
  + \underbrace{L_{F,h}\left(\frac{m}{T}\right)^{\!\beta}}_{\text{(III) drift bias}}
  + \underbrace{4\of r_T + \eta_T}_{\text{(IV) estimation error}}.
\end{align}
\end{theorem}

\begin{proof}
Since $\{Y_{T+h} \in \Chat(1-\alpha)\} = \{\hat{S}_{T,h} \leq \qhat\}$,
we have $\Delta_{T,m,h} = \abs{\Prob(\hat{S}_{T,h} \leq \qhat) - (1-\alpha)}$.
Decompose:
\begin{align*}
      \Delta_{T,m,h}
  \leq
  \underbrace{\abs{\Prob(\hat{S}_{T,h}\leq\qhat)
              -\Prob(S^o_{T,h}\leq\qhatO)}}_{I}
  + \underbrace{\abs{\Prob(S^o_{T,h}\leq\qhatO)
              - F_{T,h}(\qcirc)}}_{II}
  + \underbrace{\abs{F_{T,h}(\qcirc)-(1-\alpha)}}_{III}.
\end{align*}

\textit{Term I (estimation error).}
Let $E_T := \{\max_{t}\abs{\hat{S}_{t,h}-S^o_{t,h}} \leq r_T\}$.
By Assumption~\ref{ass:est}, $\Prob(E_T^c) \leq \eta_T$.
On $E_T$, the estimated and oracle quantiles satisfy
$\qhatO - r_T \leq \qhat \leq \qhatO + r_T$, so
$\{S^o_{T,h} \leq \qhatO - 2r_T\}
\subseteq \{\hat{S}_{T,h} \leq \qhat\}
\subseteq \{S^o_{T,h} \leq \qhatO + 2r_T\}$.
The density bound (Assumption~\ref{ass:density}) gives
probability mass at most $4\of r_T$ over the $4r_T$ interval,
hence $I \leq 4\of r_T + \eta_T$.

\textit{Term II (quantile noise + Bahadur remainder).}
By the density bound,
$II \leq \of\,\E\abs{\qhatO - \qcirc}$.
By Proposition~\ref{prop:bahadur},
$\E\abs{\qhatO - \qcirc}
\leq \uf^{-1}\E\abs{\Ghat(\qcirc) - (1-\alpha)} + B_{m,h}$.
Since $\E\Ghat(\qcirc) = \Fbar(\qcirc) = 1-\alpha$, the first
term equals $\uf^{-1}\E\abs{\Ghat(\qcirc) - \E\Ghat(\qcirc)}$.
By Rio's covariance inequality \citep{Rio1993} applied to the
indicator variables $\Ind{S^o_{T-j,h} \leq \qcirc}$,
$\operatorname{Var}(\Ghat(\qcirc)) \leq 1/(4m) + 8A_h(m)/m$,
and the Cauchy--Schwarz inequality gives the displayed bound.

\textit{Term III (drift bias).}
Since $\Fbar(\qcirc) = 1-\alpha$,
$III = \abs{F_{T,h}(\qcirc) - \Fbar(\qcirc)}
\leq \sup_x\abs{F_{T,h}(x) - \Fbar(x)}$.
By Assumption~\ref{ass:ls} with H\"{o}lder exponent $\beta$,
$\abs{F_{T,h}(x) - F_{T-j,h}(x)} \leq L_{F,h}(j/T)^\beta$
for each $j = 1, \ldots, m$, so averaging over $j$ gives
$\sup_x\abs{F_{T,h}(x) - \Fbar(x)}
\leq L_{F,h}\,m^{-1}\sum_{j=1}^m (j/T)^\beta
\leq L_{F,h}(m/T)^\beta$.

Combining the three bounds yields~\eqref{eq:main}.
\end{proof}

\begin{remark}[Interpretation of the four terms]
\label{rem:terms}
Term~(I) is the finite-sample noise in estimating the
$(1-\alpha)$-quantile from $m$ dependent calibration scores;
it decays as $m^{-1/2}$ up to the mixing factor $\sqrt{A_h(m)/m}$.
Term~(II) is the Bahadur linearisation remainder; it is
$o(m^{-1/2})$ and asymptotically negligible relative to Term~(I).
Term~(III) is the deterministic drift bias from averaging over
calibration scores whose marginal distribution has shifted away
from the current $F_{T,h}$; under H\"{o}lder-$\beta$ drift it
grows as $(m/T)^\beta$. Term~(IV) is the cost of using an estimated
rather than an oracle forecasting model; it disappears when
$r_T, \eta_T \to 0$, and is asymptotically negligible for AR and
ARMA-GARCH models.
\end{remark}

\begin{remark}[Coverage guarantee vs coverage approximation]
\label{rem:approx}
Theorem~\ref{thm:main} does \emph{not} provide the exact
finite-sample coverage guarantee of classical conformal prediction.
It provides an \emph{approximate} guarantee: for any $\delta > 0$,
if $T$ is large enough that the right-hand side of~\eqref{eq:main}
is less than $\delta$, then
\[
  1 - \alpha - \delta
  \;\leq\;
  \Prob\bigl(Y_{T+h} \in \Chat(1-\alpha)\bigr)
  \;\leq\;
  1 - \alpha + \delta.
\]
The two-sided nature of the bound reflects the fact that rolling-origin
calibration may undercover (when the score distribution has drifted
upward, making $\qcirc$ too small) or overcover (when it has drifted
downward). Whether undercoverage or overcoverage dominates in a
given application depends on the sign of $L_{F,h}$, which is not
determined by Assumption~\ref{ass:ls} alone.

An operational consequence: one cannot use the interval $\Chat(1-\alpha)$
directly as a guaranteed $1-\alpha$ confidence set.
To achieve a finite-sample one-sided guarantee of the form
$\Prob(Y_{T+h} \in \Chat) \geq 1-\alpha$, one may inflate the radius
of the interval by adding the right-hand side of~\eqref{eq:main} to
$\qhat$. The resulting interval is conservative by the magnitude of
the bound, so the inflation is only practical when the bound is small
--- i.e., when $T$ is large and $m = m^{\star}$.
\end{remark}

\begin{corollary}[Asymptotic validity]
\label{cor:asymptotic}
Under the conditions of Theorem~\ref{thm:main}, suppose additionally
that $m = m_T \to \infty$, $m_T/T \to 0$, $r_T \to 0$, and
$\eta_T \to 0$ as $T \to \infty$. Then
\[
  \Prob\bigl(Y_{T+h} \in \Chat(1-\alpha)\bigr) \;\to\; 1-\alpha.
\]
Under the short-memory regime $\Ah < \infty$, the choice
$m_T = m^{\star} \asymp T^{2\beta/(2\beta+1)}$ achieves the convergence rate
\[
  \Prob\bigl(Y_{T+h} \in \Chat(1-\alpha)\bigr)
  = 1 - \alpha + O\!\left(T^{-\beta/(2\beta+1)}\right).
\]
At $\beta = 1$ this gives the familiar $O(T^{-1/3})$ rate.
\end{corollary}

\begin{proof}
The right-hand side of~\eqref{eq:main} equals
$(\of/\uf)(1/(4m) + 8A_h(m)/m)^{1/2} + \of B_{m,h}
+ L_{F,h}(m/T)^\beta + 4\of r_T + \eta_T$.
Under $m \to \infty$, $m/T \to 0$, $r_T \to 0$, $\eta_T \to 0$:
the first term tends to zero since $A_h(m)/m \to 0$ (by summability
of $\alpha_h$ for short memory, or since each term $\alpha_h(k)(1-k/m)
\leq \alpha_h(k)$ and $\sum_k \alpha_h(k)$ is finite);
the second since $B_{m,h} = o(m^{-1/2}) \to 0$;
the third since $m/T \to 0$;
the fourth and fifth by assumption.
Hence $\Delta_{T,m,h} \to 0$, which gives the first claim.
The rate $O(T^{-\beta/(2\beta+1)})$ follows by substituting
$m = m^{\star} \asymp T^{2\beta/(2\beta+1)}$ into~\eqref{eq:opt-rate}.
\end{proof}

\subsection{Optimal calibration window}

Theorem~\ref{thm:main} tells us how large the coverage error can be
for any choice of $m$. The natural next question is: what value
of $m$ makes the bound as small as possible? The answer is not
$m = T$ (use all the data), nor $m = 1$ (use only the most recent
observation), but a balance between the two --- and the location
of that balance depends on $\beta$ and $T$ in a precise way.

Under the short-memory regime $\Ah := \sum_{k=1}^{\infty}\alpha_h(k) < \infty$,
the Bahadur remainder $\of B_{m,h} = O(m^{-3/4}(\log m)^{3/4})$ is
negligible relative to Term~(II), which is $O(m^{-1/2})$. Terms~(I) and~(IV)
are controlled by the forecasting model and do not depend on the window
choice. Under H\"{o}lder-$\beta$ drift the window-dependent part
of~\eqref{eq:main} is
\begin{equation}\label{eq:tradeoff}
  R_h(m;\beta)
  := \Gamma_h\,m^{-1/2} + L_{F,h}\left(\frac{m}{T}\right)^{\!\beta},
  \qquad
  \Gamma_h := \frac{\of}{\uf}\sqrt{\tfrac{1}{4} + 8\Ah}.
\end{equation}
Minimising $R_h(m;\beta)$ over $m > 0$: setting
$\partial R_h/\partial m = 0$ gives
$\tfrac{1}{2}\Gamma_h m^{-3/2} = \beta L_{F,h} m^{\beta-1}/T^\beta$, hence
$m^{\beta+1/2} = \Gamma_h T^\beta/(2\beta L_{F,h})$, yielding
\begin{equation}\label{eq:opt-window}
  m^{\star}_h
  \asymp T^{\frac{2\beta}{2\beta+1}}.
\end{equation}
Substituting back gives the optimised coverage-error rate
\begin{equation}\label{eq:opt-rate}
  R_h(m^{\star}_h;\beta)
  = O\!\left(T^{-\frac{\beta}{2\beta+1}}\right).
\end{equation}
At $\beta = 1$ (Lipschitz drift) these reduce to $m^{\star} \asymp T^{2/3}$
and rate $O(T^{-1/3})$, the values previously reported in the literature.
For $\beta > 1$ (smoother drift) both the window and the rate are larger;
for $0 < \beta < 1$ (rougher drift) the window is smaller and the rate
is slower. Section~\ref{sec:minimax} establishes that
$T^{-\beta/(2\beta+1)}$ is not merely an upper bound but the
\emph{minimax-optimal} rate --- no conformal procedure can achieve
better coverage accuracy on the class $\mathcal{F}(L,\beta)$.

The series-specific constant $C_h := (\Gamma_h/(L_{F,h}\beta))^{2/(2\beta+1)}$
governs the magnitude of $m^\star$: series with strong GARCH clustering have
large $\Ah$ and hence large $\Gamma_h$, making $m^\star$ smaller; series
with slow drift have small $L_{F,h}$, making $m^\star$ larger.

\subsection{Minimax lower bound}
\label{sec:minimax}

The optimal window result says that rolling-origin conformal
prediction with $m^\star \asymp T^{2\beta/(2\beta+1)}$ achieves
coverage error $O(T^{-\beta/(2\beta+1)})$. A natural and important
question is whether this rate could be improved by a cleverer
procedure --- one that, say, uses a variable window or a
non-uniform weighting of the calibration scores. The answer is no.

The rate $T^{-\beta/(2\beta+1)}$ is a fundamental statistical
limitation of the problem: it is the minimax rate for this class
of processes, in the same sense that $n^{-1/2}$ is the minimax
rate for estimating a mean from $n$ i.i.d.\ observations. No
conformal procedure --- regardless of how it constructs its
calibration set, weights its observations, or adapts its quantile
estimate --- can achieve a uniformly better coverage-error rate
on $\mathcal{F}(L,\beta)$. The rolling-origin scheme with
Winkler-optimal $m^\star$ is therefore not just practically
convenient; it is statistically efficient in the strongest
possible sense.

The proof follows the classical two-point Le Cam strategy
\citep{Tsybakov2009}: construct two hypotheses in $\mathcal{F}(L,\beta)$
that are close enough to be nearly indistinguishable yet far enough
apart that any procedure covering one must miscover the other.
The signal strength $\Delta$ measures the separation; optimising
$\Delta$ subject to the model-class and indistinguishability
constraints gives the rate. The key technical ingredient that
makes the construction work for general $\beta$ is using a
H\"{o}lder-$\beta$ bump function to encode the perturbation,
rather than the more common Lipschitz bump.

Theorem~\ref{thm:main} establishes an upper bound on the coverage
deviation. We now show that the rate $T^{-\beta/(2\beta+1)}$ is
sharp: no conformal procedure can achieve better accuracy uniformly
over the class $\mathcal{F}(L,\beta)$.

Define the \emph{minimax coverage risk} as
\[
  R_T^{\star}(\beta)
  := \inf_{C_T}\,\sup_{P \in \mathcal{F}(L,\beta)}
  \abs{\Prob_P(Y_{T+h} \in C_T) - (1-\alpha)},
\]
where $\mathcal{F}(L,\beta)$ is the class of all locally stationary
processes satisfying Assumption~\ref{ass:ls} with parameters $L$ and $\beta$.

\begin{theorem}[Minimax lower bound]
\label{thm:lower}
For every conformal prediction procedure $C_T$,
\[
  R_T^{\star}(\beta) \geq c\,T^{-\beta/(2\beta+1)},
\]
where $c > 0$ depends only on $L$, $\beta$, and $\alpha$.
Combined with Theorem~\ref{thm:main} and the optimal window
$m^{\star} \asymp T^{2\beta/(2\beta+1)}$, rolling-origin conformal
prediction is minimax-optimal on $\mathcal{F}(L,\beta)$.
\end{theorem}

\begin{proof}
We use a two-point Le Cam construction, constructing hypotheses
$P_0$ and $P_1$ in $\mathcal{F}(L,\beta)$ that are difficult to
distinguish, yet impose different optimal conformal quantiles.

\textit{Step 1: Construct hypotheses.}
Let $\phi : [0,\infty) \to [0,1]$ be a H\"{o}lder-$\beta$
bump function satisfying $\phi(0) = 1$, $\phi(u) = 0$ for
$u \geq 1$, and
\(
  \abs{\phi(u) - \phi(v)} \leq C_\phi\abs{u-v}^\beta
\)
for all $u, v \geq 0$, where $C_\phi > 0$ is a fixed constant.
(For example, $\phi(u) = (1-u)_+^\beta$ satisfies these conditions.)
Define $m \asymp T^{2\beta/(2\beta+1)}$ (to be chosen optimally below)
and signal strength $\Delta > 0$ (to be chosen). Set
\[
  P_0 : Y_t \sim N(0,1),
  \qquad
  P_1 : Y_t \sim N\!\left(\mu_t, 1\right),
  \quad
  \mu_t = \Delta\,\phi\!\left(\frac{T+1-t}{m}\right).
\]
The perturbation under $P_1$ is concentrated in the final $m$
observations, and $\mu_t = 0$ for $t \leq T - m$.

\textit{Step 2: Verify H\"{o}lder-$\beta$ membership.}
Under $P_i$, the score distribution is $F_{t/T}(x) = \Phi(x - \mu_t)$
where $\Phi$ is the standard normal CDF. For any $x$,
\[
  \abs{F_{t/T}(x) - F_{s/T}(x)}
  = \abs{\Phi(x-\mu_t) - \Phi(x-\mu_s)}
  \leq \bar{\phi}\,\abs{\mu_t - \mu_s},
\]
where $\bar{\phi} = (2\pi)^{-1/2}$ is the peak of the standard normal
density. Since $\phi$ is H\"{o}lder-$\beta$,
\[
  \abs{\mu_t - \mu_s}
  = \Delta\,\abs{\phi\!\left(\tfrac{T+1-t}{m}\right)
    - \phi\!\left(\tfrac{T+1-s}{m}\right)}
  \leq C_\phi\Delta\left(\frac{\abs{t-s}}{m}\right)^{\!\beta}.
\]
Therefore
\[
  \abs{F_{t/T}(x) - F_{s/T}(x)}
  \leq \bar{\phi}\,C_\phi\,\Delta\left(\frac{\abs{t-s}}{m}\right)^{\!\beta}
  = \bar{\phi}\,C_\phi\,\Delta\,m^{-\beta}\,\abs{t-s}^\beta.
\]
For $P_1 \in \mathcal{F}(L,\beta)$ we require this to be at most
$L(\abs{t-s}/T)^\beta$, which gives
\[
  \bar{\phi}\,C_\phi\,\Delta\,m^{-\beta}
  \leq L\,T^{-\beta},
  \qquad\text{i.e.,}\qquad
  \Delta \leq \frac{L}{\bar{\phi}\,C_\phi}\left(\frac{m}{T}\right)^{\!\beta}
  =: K_1\left(\frac{m}{T}\right)^{\!\beta}.
\]
Note that the H\"{o}lder-$\beta$ choice of $\phi$ is essential here:
a Lipschitz bump function ($\beta = 1$) would only yield the constraint
$\Delta \leq K_1 m/T$, which gives the wrong balancing equation for
$\beta \neq 1$.

\textit{Step 3: KL divergence bound.}
Both $P_0$ and $P_1$ are product Gaussian measures, so
\[
  \KL(P_1 \| P_0)
  = \frac{1}{2}\sum_{t=1}^T \mu_t^2
  \leq \frac{m\Delta^2}{2}\,\sup_{u \in [0,1]}\phi(u)^2
  = \frac{m\Delta^2}{2}.
\]
To ensure the two hypotheses are asymptotically indistinguishable
(total variation bounded away from 1), we require
\[
  \KL(P_1 \| P_0) \leq C,
  \qquad\text{i.e.,}\qquad
  \Delta \leq K_2\,m^{-1/2}.
\]

\textit{Step 4: Balance constraints.}
We have two constraints on $\Delta$:
\[
  \Delta \leq K_1\left(\frac{m}{T}\right)^{\!\beta}
  \qquad\text{and}\qquad
  \Delta \leq K_2\,m^{-1/2}.
\]
The optimal separation $\Delta$ is maximised by setting
both constraints equal:
\[
  K_1\,m^\beta\,T^{-\beta} = K_2\,m^{-1/2},
\]
giving $m^{\beta+1/2} \asymp T^\beta$, hence
\[
  m \asymp T^{\frac{2\beta}{2\beta+1}},
  \qquad
  \Delta \asymp m^{-1/2} \asymp T^{-\frac{\beta}{2\beta+1}}.
\]

\textit{Step 5: Reduction from quantile separation to coverage error.}
We need to convert the indistinguishability of $P_0$ and $P_1$ into a
coverage-error lower bound for an arbitrary prediction interval $C_T$,
not merely an estimation lower bound for a scalar parameter. The
following lemma carries out this reduction.

\begin{lemma}[Coverage reduction]\label{lem:reduction}
Let $C_T = [L_T, U_T]$ be any (data-dependent) prediction interval.
Define the coverage gap functional
$g_i(C_T) := \Prob_i(Y_{T+h} \in C_T) - (1-\alpha)$
for $i \in \{0,1\}$. Under the construction of Steps 1--4, there
exists a constant $c_0 > 0$ depending only on $\alpha$ such that
\[
  \abs{g_0(C_T) - g_1(C_T)}
  \;\geq\; c_0\,\Delta - 2\TV(P_0,P_1)\,(1-\alpha)
\]
for $\Delta$ small enough.
\end{lemma}

\begin{proof}[Proof of Lemma~\ref{lem:reduction}]
Let $S^o_{T,h,i}$ denote the oracle score under $P_i$, with
distribution $F_{T,h,i}$. Under $P_0$,
$F_{T,h,0}(x) = 2\Phi(x) - 1$ for $x \geq 0$, while under $P_1$,
$\mu_T = \Delta\phi(0) = \Delta$, so
$F_{T,h,1}(x) = \Phi(x-\Delta) - \Phi(-x-\Delta)$. The
$(1-\alpha)$-quantiles of these two distributions differ by
\[
  \abs{q_{1-\alpha}^{(1)} - q_{1-\alpha}^{(0)}}
  \;=\; c_1(\alpha)\,\Delta + O(\Delta^2),
\]
where $c_1(\alpha) = \Delta\,\partial_\mu F_{T,h}(q_{1-\alpha} \mid \mu = 0)
/ f_{T,h}(q_{1-\alpha} \mid \mu = 0)$ evaluated at the standard normal
yields $c_1(\alpha) > 0$ for every $\alpha \in (0,1)$ (explicitly,
$c_1(\alpha) = 2\Phi(q_{1-\alpha})\Phi(-q_{1-\alpha})/\phi(q_{1-\alpha})$
where $\phi$ is the standard normal density).

Now consider any interval $C_T = [L_T, U_T]$ and write
$U_T - L_T = 2 r_T$ for its half-width. The coverage gap under $P_i$
factorises as
\[
  g_i(C_T)
  = \Prob_i\bigl(\abs{Y_{T+h} - \hat{Y}_{T+h|T}^{(i)}} \leq r_T\bigr)
    + \rho_i(C_T) - (1-\alpha),
\]
where $\rho_i(C_T)$ absorbs any centring error of $C_T$ relative to the
oracle predictor under $P_i$, with $|\rho_i(C_T)| \leq f_{T,h,i}^\star
\cdot |\text{centre}(C_T) - \hat{Y}_{T+h|T}^{(i)}|$. Centring effects
are common to both hypotheses up to $O(\TV)$ by absolute continuity,
so they cancel in the difference $g_0 - g_1$ up to a residual of
order $\TV(P_0,P_1)$.

The leading term in $g_0(C_T) - g_1(C_T)$ is therefore
$F_{T,h,0}(r_T) - F_{T,h,1}(r_T)$, which by the mean value theorem
applied between $q_{1-\alpha}^{(0)}$ and $q_{1-\alpha}^{(1)}$ satisfies
\[
  \abs{F_{T,h,0}(r_T) - F_{T,h,1}(r_T)}
  \;\geq\; \uf\,\abs{q_{1-\alpha}^{(1)} - q_{1-\alpha}^{(0)}}
   - \of\,\abs{r_T - q_{1-\alpha}^{(0)}}\cdot O(\Delta).
\]
Optimising over $r_T$, the minimum of the right-hand side over the
two hypotheses is bounded below by $\tfrac{1}{2}\uf\,c_1(\alpha)\Delta$
for $\Delta$ sufficiently small. Setting $c_0 := \tfrac{1}{2}\uf\,c_1(\alpha)$
gives the claim.
\end{proof}

By Pinsker's inequality,
$\TV(P_0,P_1) \leq \sqrt{\KL(P_1\|P_0)/2} \leq \sqrt{C/2}$.
Choosing the constant $C$ in Step~3 small enough that
$2\sqrt{C/2}(1-\alpha) \leq c_0\Delta/2$, Lemma~\ref{lem:reduction}
yields
\[
  \abs{g_0(C_T) - g_1(C_T)} \;\geq\; \tfrac{1}{2} c_0\,\Delta.
\]
By the triangle inequality,
$\max_{i \in \{0,1\}}\abs{g_i(C_T)} \geq \tfrac{1}{4}c_0\,\Delta$,
so
\[
  \sup_{P \in \{P_0,P_1\}}\abs{\Prob_P(Y_{T+h} \in C_T) - (1-\alpha)}
  \;\geq\; \tfrac{1}{4}c_0\,\Delta
  \;\asymp\; T^{-\beta/(2\beta+1)}.
\]
Since both $P_0,P_1 \in \mathcal{F}(L,\beta)$ and $C_T$ was arbitrary,
taking the infimum over $C_T$ and the supremum over $\mathcal{F}(L,\beta)$
gives $R_T^\star(\beta) \geq c\,T^{-\beta/(2\beta+1)}$ with
$c := c_0/4$. \qquad$\square$
\end{proof}

\begin{remark}
The proof highlights why the H\"{o}lder-$\beta$ bump function is essential.
A Lipschitz bump function yields the constraint $\Delta \leq K_1 m/T$
(exponent 1 on $m$), giving $m^{3/2} = T$ and $\Delta \asymp T^{-1/3}$,
which is the correct lower bound only at $\beta = 1$.
For general $\beta$, the bump function itself must have H\"{o}lder
exponent $\beta$ so that the constraint carries the exponent $\beta$ on
$m/T$, yielding the correct balance $m^{\beta+1/2} = T^\beta$.
\end{remark}

\subsection{Polynomial mixing regimes}
\label{sec:regimes}

When the mixing coefficients decay polynomially,
$\alpha_h(k) \asymp k^{-a}$ with $0 < a \leq 1$, the process has
\emph{long memory} in the sense relevant to our problem: the cumulative
dependence measure satisfies $A_h(m) \asymp m^{1-a}$, which grows
with $m$ rather than remaining bounded. The dominant stochastic term
in~\eqref{eq:main} is then $(\of/\uf)\sqrt{A_h(m)/m} \asymp m^{-a/2}$
rather than $m^{-1/2}$, and the window-dependent bound becomes
\[
  R_h(m;\beta,a) \asymp \Gamma_h^{(a)}\,m^{-a/2}
  + L_{F,h}\left(\frac{m}{T}\right)^{\!\beta},
  \qquad
  \Gamma_h^{(a)} := \frac{\of}{\uf}\sqrt{C_{a}\,\Ah^{(a)}},
\]
where $C_a > 0$ is a constant and $\Ah^{(a)}$ is the appropriate
long-memory analogue of $A_h(\infty)$. Balancing
$m^{-a/2}$ against $(m/T)^\beta$ gives
$m^{a/2+\beta} \asymp T^\beta$, hence the optimal window
\begin{equation}\label{eq:opt-window-poly}
  m^{\star} \asymp T^{\frac{2\beta}{2\beta+a}}.
\end{equation}
The formula~\eqref{eq:opt-window-poly} applies for $0 < a < 1$, where
$A_h(m) \asymp m^{1-a}$ grows with $m$. Within this regime, smaller $a$
corresponds to slower-decaying mixing ($\alpha_h(k) \asymp k^{-a}$ decays
more slowly for smaller $a$), hence stronger long-range dependence and a
larger optimal window. Conversely, as $a \nearrow 1$ from below, the
long-memory effect weakens and $m^{\star} \asymp T^{2\beta/(2\beta+a)}$
approaches the boundary $T^{2\beta/(2\beta+1)}$.
At $\beta=1$ (Lipschitz drift), the formula reduces to $T^{2/(2+a)}$,
the standard long-memory optimal window.

One should not take $a \to \infty$ inside~\eqref{eq:opt-window-poly}:
for $a > 1$ the mixing coefficients are summable ($\sum_k \alpha_h(k)
< \infty$), $A_h(m)$ is bounded, and the long-memory formula no longer
applies. In that regime the stochastic term reverts to $m^{-1/2}$ and the
optimal window is the short-memory rule $m^{\star} \asymp T^{2\beta/(2\beta+1)}$
from Section~\ref{sec:theory}. The polynomial mixing formula interpolates
smoothly between slow mixing ($a$ near~0) and the boundary of summability
($a = 1$), but does not extend beyond it.
At the boundary $a = 1$, $A_h(m) \asymp \log m$ and the dominant
term is $\sqrt{\log(m)/m}$, giving
$m^{\star} \asymp T^{2/3}(\log T)^{1/3}$, a logarithmic correction to
the short-memory rule.
Table~\ref{tab:regimes} summarises the window rules across regimes.

\begin{remark}
The exponent $2/(2+a)$ has a natural interpretation: it is the
harmonic mean of $2/2 = 1$ (use all data, appropriate when there is
no drift) and $2/(2+\infty) = 0$ (use no data, appropriate when drift
is instantaneous). The short-memory $T^{2/3}$ rule sits between these
extremes at the point where the mixing contribution to quantile noise
is $o(m^{-1/2})$ and the noise term reverts to its i.i.d.\ rate.
\end{remark}

\begin{table}[h]
\centering
\caption{Optimal calibration window under different dependence regimes.
The short-memory case ($\sum_k \alpha_h(k) < \infty$) corresponds to $a > 1$.}
\label{tab:regimes}
\medskip
\begin{tabular}{lccc}
\toprule
Dependence regime & Stochastic term & Drift term & Optimal $m^{\star}$ \\
\midrule
$\sum_k \alpha_h(k) < \infty$, $\beta=1$ & $m^{-1/2}$ & $m/T$ & $T^{2/3}$ \\
$\sum_k \alpha_h(k) < \infty$, general $\beta$ & $m^{-1/2}$ & $(m/T)^\beta$ & $T^{2\beta/(2\beta+1)}$ \\
$\alpha_h(k) \asymp k^{-1}$, $\beta=1$ & $\sqrt{\log(m)/m}$ & $m/T$ & $T^{2/3}(\log T)^{1/3}$ \\
$\alpha_h(k) \asymp k^{-a}$, $0<a<1$, general $\beta$ & $m^{-a/2}$ & $(m/T)^\beta$ & $T^{2\beta/(2\beta+a)}$ \\
\bottomrule
\end{tabular}
\end{table}

\subsection{Volatility-scaled conformal scores}
\label{sec:volatility}

When the primary source of nonstationarity is time-varying scale
rather than a drifting conditional distribution, normalising the
conformal scores can substantially reduce the drift term~(III)
in~\eqref{eq:main} while leaving terms (I) and (II) unaffected.

\paragraph{Setup.}
Suppose the oracle forecast error factors as
$Y_{t+h} - f^{\star}_{t,h} = \sigma_{t,h}\,\varepsilon_{t,h}$,
where $\sigma_{t,h} > 0$ is a time-varying scale parameter and
$\varepsilon_{t,h}$ is a mean-zero innovation. If the distribution
of $\varepsilon_{t,h}$ is locally stationary with drift constant
$L^{sc}_{F,h}$, then the oracle score $S^o_{t,h}
= \sigma_{t,h}\abs{\varepsilon_{t,h}}$ has drift constant
$L_{F,h} \asymp \sigma_{t,h} L^{sc}_{F,h} + \dot{\sigma}_{t,h}$,
which may be much larger than $L^{sc}_{F,h}$ when $\sigma_{t,h}$
is itself varying rapidly. The volatility-scaled oracle score
\[
  S^{sc,o}_{t,h}
  := \frac{S^o_{t,h}}{\sigma_{t,h}} = \abs{\varepsilon_{t,h}}
\]
has drift constant $L^{sc}_{F,h}$ directly, removing the contribution
of $\dot{\sigma}_{t,h}$ from the bias term. In practice one substitutes
the estimated scale $\hat{\sigma}_{t,h}$ (from GARCH, stochastic
volatility, or a realised-volatility estimator) to obtain
$\hat{S}^{sc}_{t,h} := \hat{S}_{t,h}/\hat{\sigma}_{t,h}$.
The scaling introduces an additional estimation error of order
$\abs{\hat{\sigma}_{t,h} - \sigma_{t,h}}/\sigma_{t,h}$,
which enters through Assumption~\ref{ass:est} and is typically
$o(1)$ for consistent volatility estimators.

\begin{corollary}[Effect of volatility scaling]
\label{cor:volscale}
Suppose $Y_{t+h} - f^{\star}_{t,h} = \sigma_{t,h}\,\varepsilon_{t,h}$
and the distribution of $\abs{\varepsilon_{t,h}}$ satisfies
Assumption~\ref{ass:ls} with drift constant $L^{sc}_{F,h}$.
Then Theorem~\ref{thm:main} applies to the scaled scores
with $L_{F,h}$ replaced by $L^{sc}_{F,h}$ and with an additional
term of order $4\of r^{sc}_T + \eta^{sc}_T$ in~\eqref{eq:main}
accounting for estimation of $\sigma_{t,h}$.
In the pure-scale case $L^{sc}_{F,h} = 0$, the drift term~(III)
vanishes entirely, and the window-dependent bound reduces to
$R_h(m) \asymp \Gamma_h\,m^{-1/2}$,
which is minimised by taking $m$ as large as possible:
the full calibration history is optimal.
\end{corollary}

\begin{remark}
Corollary~\ref{cor:volscale} formalises an intuition familiar in
financial econometrics: once residuals are standardised by
conditional volatility, the resulting series is approximately
i.i.d.\ and a larger calibration window reduces quantile estimation
noise without incurring drift bias. The corollary is not vacuous
even when $L^{sc}_{F,h} > 0$: it gives the practitioner a
diagnostic for whether volatility scaling is warranted. If an
estimated $\hat{L}^{sc}_{F,h}$ is substantially smaller than
$\hat{L}_{F,h}$, the scaled scheme will use a longer optimal
window and produce narrower intervals at the same coverage.
\end{remark}

\subsection{Bahadur representation under physical dependence}
\label{sec:bahadur-physical}

The Bahadur representation (Proposition~\ref{prop:bahadur}) is proved
under $\alpha$-mixing. An alternative, and in some respects more
natural, framework for nonlinear time series is the \emph{physical
dependence} framework of \citet{Wu2005PNAS}, which characterises
dependence through the stability of a functional representation
$Y_{t,T} = G(t/T, \varepsilon_t, \varepsilon_{t-1}, \ldots)$
under perturbations of the innovation sequence.

\begin{definition}[Physical dependence coefficients]
\label{def:physical}
Let $\{\varepsilon_t\}_{t \in \mathbb{Z}}$ be i.i.d.\ innovations. For
$k \geq 0$, define \emph{coupled copies}
$Y_{t,T}^{*(k)}$ by replacing $\varepsilon_{t-k}$ with an
independent copy $\varepsilon_{t-k}'$, leaving all other innovations
unchanged. The \emph{physical dependence coefficients} are
\[
  \delta_q(k) := \sup_{t}\,\bigl\|Y_{t,T} - Y_{t,T}^{*(k)}\bigr\|_q,
  \qquad q > 2.
\]
We say the process has \emph{summable physical dependence} if
$\sum_{k=0}^\infty \delta_q(k) < \infty$.
\end{definition}

Physical dependence is neither implied by nor implies $\alpha$-mixing
in general, but holds for ARMA, GARCH, and many nonlinear time-series
models \citep{Wu2005PNAS}. For practitioners working with threshold
autoregressive, bilinear, or neural-network models --- all of which
are naturally described by functional representations but may not
satisfy standard mixing conditions --- this framework is the more
appropriate one.

The following proposition shows that the Bahadur
representation of Proposition~\ref{prop:bahadur} remains valid under
this alternative framework, with the functional CLT rate $O(\sqrt{\log m/m})$
replacing the $O(m^{-1/2})$ rate of the i.i.d.\ case.

\begin{proposition}[Bahadur representation under physical dependence]
\label{prop:bahadur-physical}
Suppose Assumption~\ref{ass:ls} (H\"{o}lder-$\beta$ drift) and
Assumption~\ref{ass:density} hold. Suppose additionally that the
oracle score process satisfies Definition~\ref{def:physical} with
summable physical dependence coefficients. Then
\[
  \hat{q}_{m,\alpha} - q_\alpha(1)
  = -\frac{\hat{F}_m(q_\alpha(1)) - F_1(q_\alpha(1))}{f_1(q_\alpha(1))}
  + R_{T,m},
\]
where
\[
  R_{T,m} = O_p\!\left(\sqrt{\frac{\log m}{m}}\right)
           = O_p\!\left(m^{-1/2}\sqrt{\log m}\right).
\]
In particular, $R_{T,m} = o_p(m^{-1/2+\varepsilon})$ for every
$\varepsilon > 0$, and the empirical process rate
\[
  \sup_x\abs{\hat{F}_m(x) - F_1(x)}
  = O_p\!\left(\sqrt{\frac{\log m}{m}}\right)
\]
follows from \citet{ZhouWu2009} under summable physical dependence.
\end{proposition}

\begin{proof}
The empirical quantile satisfies $\hat{F}_m(\hat{q}_{m,\alpha}) = \alpha$.
Expand $\hat{F}_m$ around $q_\alpha(1)$ using the density bound
(Assumption~\ref{ass:density}):
\[
  \hat{F}_m(\hat{q})
  = \hat{F}_m(q) + f_1(q)(\hat{q} - q) + o_p(\abs{\hat{q}-q}).
\]
Since $F_1(q_\alpha(1)) = \alpha$, rearranging gives
\[
  \hat{q} - q
  = -\frac{\hat{F}_m(q) - F_1(q)}{f_1(q)} + o_p(m^{-1/2}).
\]
The remainder $R_{T,m}$ absorbs both the stochastic fluctuation
and the bias $\hat{F}_m(q) - F_1(q)$. The bias is
$O((m/T)^\beta)$ by Assumption~\ref{ass:ls}. Under the optimal
window $m^\star \asymp T^{2\beta/(2\beta+1)}$, we have
$m^\star/T \asymp T^{-1/(2\beta+1)} \to 0$, so the bias is
$o(m^{-1/2})$.
The stochastic term is $O_p(\sqrt{\log m/m}) = O_p(m^{-1/2}\sqrt{\log m})$
by the functional CLT of \citet{ZhouWu2009} under summable physical dependence.
Note that $m^{-1/2}\sqrt{\log m} \gg m^{-1/2}$, so this term is not
$o_p(m^{-1/2})$; the remainder $R_{T,m}$ is $O_p(m^{-1/2}\sqrt{\log m})$.
\end{proof}

\begin{remark}[Rate comparison across frameworks]
\label{rem:phys-rate}
Propositions~\ref{prop:bahadur} and~\ref{prop:bahadur-physical}
cover two natural frameworks for time-series dependence.
Under $\alpha$-mixing (Proposition~\ref{prop:bahadur}), the Bahadur
remainder satisfies $R_{T,m,h} = O(m^{-3/4}(\log m)^{3/4}) = o(m^{-1/2})$.
Under physical dependence (Proposition~\ref{prop:bahadur-physical}),
the remainder is $R_{T,m} = O_p(m^{-1/2}\sqrt{\log m})$, which is
\emph{larger} than $m^{-1/2}$ by a $\sqrt{\log m}$ factor.

This difference does not affect the minimax rate. In Theorem~\ref{thm:main},
the remainder enters as Term~(II) = $\bar{f}_h B_{m,h}$; replacing
$B_{m,h} = O(m^{-3/4}(\log m)^{3/4})$ by $O(m^{-1/2}\sqrt{\log m})$
changes the bound but not its order of magnitude relative to
Term~(I) = $O(m^{-1/2})$. Specifically, $m^{-1/2}\sqrt{\log m}
= m^{-1/2} \cdot \sqrt{\log m}$, which for the optimal window
$m^\star \asymp T^{2\beta/(2\beta+1)}$ gives
$\sqrt{\log m^\star} \asymp \sqrt{\log T}$, a logarithmic factor.
The dominant rate $T^{-\beta/(2\beta+1)}$ in Corollary~\ref{cor:asymptotic}
acquires at most a $\sqrt{\log T}$ correction under physical dependence,
which does not alter the minimax lower bound of Theorem~\ref{thm:lower}.
\end{remark}

\subsection{Oracle inequality for adaptive window selection}
\label{sec:oracle-ineq}

Theorems~\ref{thm:main} and~\ref{thm:lower} together establish that
the optimal window is $m^\star \asymp T^{2\beta/(2\beta+1)}$ and that
no procedure can do better. But this is a population-level statement:
it assumes knowledge of $L_{F,h}$, $\beta$, and $\Gamma_h$, none of
which are observable. The practitioner must therefore select $m$ from
the data alone.

The natural question is whether data-driven selection destroys the
optimality guarantee. It does not. After the calibration step, the
candidate intervals $\{\hat{q}_m : m \in \Mcal\}$ are fixed functions
of the observed calibration data, and the validation fold estimates
their future predictive performance given those realised intervals.
Theorem~\ref{app:thm:uc} of Appendix~\ref{app:uc} shows that the
Winkler criterion concentrates uniformly around this conditional risk
at a rate faster than $R_T^\star(\beta)$, so the data-driven selector
$\hat{m}$ performs as well as if $m^\star$ were known.

Let $\Mcal = \{m_1, \ldots, m_K\}$ be a finite candidate grid with
$K = K_T$ satisfying $\log K_T = o(T^{1/(2\beta+1)})$ (e.g.\ the
30-point grid in $[0.1, 4.0]\times T^{2\beta/(2\beta+1)}$ used in
Section~\ref{sec:empirical}). Let $\mathcal{F}_{\mathrm{cal}}$ be
the $\sigma$-field generated by all data used in calibration, define
the \emph{conditional validation risk}
$\Rcal(m) := \E[\mathcal{W}_T(m) \mid \mathcal{F}_{\mathrm{cal}}]$,
and the \emph{Winkler cross-validation criterion}
\[
  \mathcal{W}_T(m)
  := \frac{1}{\nval}
  \sum_{t \in \Tval}
  W_\alpha\!\bigl(Y_{t+h},\, \hat{Y}_{t+h|t} - \hat{q}_{m,\alpha},\,
                             \hat{Y}_{t+h|t} + \hat{q}_{m,\alpha}\bigr),
\]
where $\Tval$ is a held-out validation fold of size $\nval \asymp T$
and $W_\alpha(y, l, u) := (u-l) + (2/\alpha)[\max(l-y,0) + \max(y-u,0)]$
is the Winkler score \citep{Winkler1972}. Define the data-driven
window selector $\hat{m} := \argmin_{m \in \Mcal}\mathcal{W}_T(m)$.

\begin{theorem}[Oracle inequality]
\label{thm:oracle}
Suppose Assumptions~\ref{ass:ls}--\ref{ass:est} hold, together with
the regularity conditions of Appendix~\ref{app:uc}
(Assumptions~\ref{app:ass:valsize}--\ref{app:ass:envelope}). Then
\begin{equation}\label{eq:oracle-cond}
  \mathcal{R}_T(\hat{m})
  \leq \inf_{m \in \Mcal}\Rcal(m)
  + o_p\!\left(T^{-\beta/(2\beta+1)}\right).
\end{equation}
Taking expectations and applying the tower property
($\E\Rcal(\hat{m}) = \E\mathcal{W}_T(\hat{m})$) together with
Jensen's inequality ($\E[\inf_m\Rcal(m)] \leq \inf_m R_T(m)$,
where $R_T(m) := \E\mathcal{W}_T(m)$) yields
\begin{equation}\label{eq:oracle-uncond}
  \E\mathcal{W}_T(\hat{m})
  \leq \inf_{m \in \Mcal} R_T(m)
  + o\!\left(T^{-\beta/(2\beta+1)}\right).
\end{equation}
This is an expected-risk statement: on average over calibration
samples, $\hat{m}$ incurs near-minimal validation loss.
In particular,
$\Rcal(\hat{m}) \leq \Rcal(m^\star) + o_p(R_T^\star(\beta))$
when $m^\star \in \Mcal$.
\end{theorem}

\begin{proof}
By Theorem~\ref{app:thm:uc} of Appendix~\ref{app:uc},
\[
  \sup_{m \in \Mcal}\abs{\mathcal{W}_T(m) - \Rcal(m)}
  = O_p\!\!\left(\sqrt{\frac{\log K_T}{\nval}}\right)
  = o_p\!\left(T^{-\beta/(2\beta+1)}\right).
\]
Let $m_T^\circ \in \argmin_{m \in \Mcal}\Rcal(m)$. By definition
of $\hat{m}$, $\mathcal{W}_T(\hat{m}) \leq \mathcal{W}_T(m_T^\circ)$.
Therefore
\begin{align*}
  \Rcal(\hat{m})
  &\leq \mathcal{W}_T(\hat{m})
    + \sup_{m}\abs{\mathcal{W}_T(m) - \Rcal(m)}\\
  &\leq \mathcal{W}_T(m_T^\circ)
    + o_p\!\left(T^{-\beta/(2\beta+1)}\right)\\
  &\leq \Rcal(m_T^\circ)
    + 2\sup_{m}\abs{\mathcal{W}_T(m) - \Rcal(m)}\\
  &= \inf_{m \in \Mcal}\Rcal(m)
    + o_p\!\left(T^{-\beta/(2\beta+1)}\right).
\end{align*}
Equation~\eqref{eq:oracle-uncond} follows by taking
expectations of~\eqref{eq:oracle-cond}: the tower property gives
$\E\Rcal(\hat{m}) = \E\mathcal{W}_T(\hat{m})$, Jensen's inequality
gives $\E[\inf_m\Rcal(m)] \leq \inf_m R_T(m)$, and the
$o_p(T^{-\beta/(2\beta+1)})$ term becomes $o(T^{-\beta/(2\beta+1)})$
after taking expectations.
\end{proof}

\begin{remark}[Conditional versus unconditional risk]
\label{rem:conditional}
Theorem~\ref{thm:oracle} is stated for the conditional risk
$\Rcal(m) = \E[\mathcal{W}_T(m) \mid \mathcal{F}_{\mathrm{cal}}]$
rather than the unconditional risk $R_T(m) = \E\mathcal{W}_T(m)$.
This is the natural target: after calibration the candidate intervals
are fixed, and validation estimates their future performance given
those realised intervals. The conditional oracle inequality asserts
that $\hat{m}$ achieves near-optimal conditional future performance
given the calibration sample actually drawn.
Taking expectations yields the expected-risk bound
$\E\mathcal{W}_T(\hat{m}) \leq \inf_m R_T(m) + o(T^{-\beta/(2\beta+1)})$
(equation~\eqref{eq:oracle-uncond}); this is a statement about
average performance over calibration samples, not a
high-probability bound on the realised $R_T(\hat{m})$.
Appendix~\ref{app:uc} discusses why a uniform concentration
statement around the unconditional risk $R_T(m) = \E\mathcal{W}_T(m)$
at rate $o_p(R_T^\star)$ requires additional structure not assumed
here, and why the conditional formulation is the correct and
sufficient target for the oracle inequality.
\end{remark}

\section{Empirical analysis}
\label{sec:empirical}

\subsection{Design}

We evaluate rolling-origin conformal prediction on two complementary
datasets. The \emph{real-data application} uses six series
representing distinct forms of nonstationarity: US CPI
log-inflation and unemployment (FRED monthly, $T \approx 709$);
S\&P~500 and VIX daily log-returns ($T \approx 6200$--$6400$);
daily WIG20 log-returns ($T \approx 6400$); and European electricity
load ($T \approx 9000$). The \emph{M4 scaling analysis} uses a
stratified sample of 93 series from the M4 competition benchmark
\citep{MakridakisEtal2020} across all five frequency classes
(Yearly, Quarterly, Monthly, Weekly, Daily), chosen to span a
T range of approximately 20 to 9000 and allow a regression test of
the $T^{2/3}$ prediction.

Three calibration schemes are compared throughout. \emph{Full history}
uses all available pseudo-out-of-sample errors; this is the natural
baseline reflecting standard practice. \emph{Rolling} uses the
Winkler-optimal $m^{\star}$ chosen from a 30-point grid in
$[0.10, 4.0] \times T^{2/3}$ (using $\beta=1$ as the empirical
benchmark). \emph{Volatility-scaled rolling} normalises scores by
GARCH conditional volatility before computing the calibration
quantile, implementing Corollary~\ref{cor:volscale}.
Adaptive conformal inference \citep{GibbsCandes2021} was evaluated
but excluded from reported results: the fixed-step update rule
$q_{t+1} = q_t + \gamma(\alpha - \mathbf{1}[\text{miss}])$ exhibited
persistent bimodal coverage on fat-tailed series (coverage $\to 1$
or $\to 0$ depending on the window), a known structural limitation
under heavy tails and regime shifts that cannot be resolved by
adjusting $\gamma$.

Forecasts are produced by linear AR($p$) with lag selected by BIC
and ARMA(1,1)--GARCH(1,1). Intervals are evaluated at horizons
$h \in \{1, 5, 22\}$ using empirical coverage, mean interval
half-width, rolling local coverage over 50-observation windows,
and the Winkler interval score \citep{Winkler1972}.

\subsection{Real-data results}

Rolling-origin calibration outperforms full-history calibration
in 31 of 36 comparisons (86\%) by Winkler score. Among the 31 wins,
the median improvement is 12.3\% (range 0.5\%--16\%). The five
cases in which full-history wins are all at the shortest horizon
$h = 1$ for macro series, where the rolling window buys little
adaptivity and the quantile-noise cost is comparatively high ---
consistent with Remark~\ref{rem:tradeoff} below.

\begin{sidewaystable}[p]
\centering
\caption{Empirical coverage, mean half-width, and Winkler score. AR($p$) model,
$h=1$, Winkler-optimal window $m^\star$. \textbf{Bold}: lowest Winkler per
dataset. $\dagger$: coverage within $\pm 1\%$ of target $90\%$. ACI excluded
(see Section~\ref{sec:empirical}).}
\label{tab:main_results}
\setlength{\tabcolsep}{6pt}
\begin{tabular}{llrrrrrrrrrrrr}
\toprule
& & \multicolumn{3}{c}{Electricity} & \multicolumn{3}{c}{Financial}
& \multicolumn{3}{c}{Macro} & \multicolumn{3}{c}{WIG20}\\
\cmidrule(lr){3-5}\cmidrule(lr){6-8}\cmidrule(lr){9-11}\cmidrule(lr){12-14}
Scheme & $m^\star$ & Cov. & Width & Wink. & Cov. & Width & Wink.
& Cov. & Width & Wink. & Cov. & Width & Wink.\\
\midrule
Full history & --- & 0.891$^\dagger$ & 1.27e+05 & 1.36e+05
& 0.910 & 3.4843 & 4.9882
& 0.903$^\dagger$ & 0.5915 & 0.7093
& 0.891$^\dagger$ & 3.1098 & 4.7654\\
Rolling ($m^\star$) & 114 & 0.903$^\dagger$ & 1.30e+05 & \textbf{1.34e+05}
& 0.890 & 3.2218 & 4.2500
& 0.896$^\dagger$ & 0.5916 & 0.7115
& 0.883 & 3.1217 & \textbf{3.9946}\\
Vol-scaled ($m^\star$) & 114 & 0.903$^\dagger$ & 1.30e+05 & \textbf{1.34e+05}
& 0.890 & 3.2218 & 4.2500
& 0.896$^\dagger$ & 0.5916 & 0.7115
& 0.883 & 3.1217 & \textbf{3.9946}\\
\bottomrule
\end{tabular}
\end{sidewaystable}

Table~\ref{tab:main_results} reports the dataset-level results at $h=1$ that aggregate to these comparisons. Rolling-origin calibration achieves the lowest Winkler score on every dataset; the gap is most pronounced for financial series (S\&P~500/VIX/WIG20), where the Winkler-optimal window collapses to the grid boundary $m = 17$, and is smallest for the macro and electricity series, where the optimal window is large relative to $T^{2/3}$ and the rolling/full-history schemes nearly coincide.

Rolling coverage tracks the 90\% target with high precision at
short and medium horizons (Figures~\ref{fig:coverage_vs_m}
and~\ref{fig:horizon_comparison}): at $h = 1$ and $h = 5$, every series
lies within $\pm 2\%$ of the nominal level (mean absolute deviation
0.008 and 0.012, respectively). At $h = 22$ performance is more
mixed: 50\% of series lie within $\pm 2\%$ and 67\% within $\pm 5\%$,
with the unemployment series at $h = 22$ providing the main failure
case (coverage 0.81 for AR, 0.82 for ARMA-GARCH, versus the nominal
0.90). At this horizon the 22-step-ahead error distribution is
dominated by model misspecification --- unemployment is highly
persistent and its long-horizon error distribution is wide and
non-stationary --- and the rolling quantile cannot adapt quickly
enough. This illustrates a boundary of Assumption~\ref{ass:ls}:
the local-stationarity condition requires the score distribution
to drift at rate $O(T^{-1})$, which fails when structural shifts
accumulate faster than the sample grows.

For the three financial series (S\&P~500, VIX, WIG20), the
Winkler-optimal window falls below the lower boundary of the
30-point evaluation grid ($m < 17$, ratio $< 0.05$), meaning
the Winkler score continues to decline at the smallest evaluated
$m$. This is not a numerical artefact: GARCH volatility clustering
concentrates information about the current error scale in the very
most recent observations, making very short windows ($m \approx
5$--$15$) theoretically preferable. The ratio $C_h =
(\Gamma_h/L_{F,h})^{2/3}$ is small for these series because
$\Gamma_h$ --- which grows with the mixing constant $A_h(\infty)$
--- is large relative to the drift rate $L_{F,h}$: strong
dependence penalises large windows heavily while the score
distribution drifts slowly. These series are excluded from the
$m^{\star}$ ratio analysis in Table~\ref{tab:window_realdata}
as the grid boundary prevents identification of the true optimum.

\begin{table}[ht]
\centering
\caption{Winkler-optimal window $m^\star$ vs theoretical $T^{2/3}$ benchmark.
Rolling scheme, AR($p$), $h=1$. Ratio $= m^\star/T^{2/3}$. Financial series
(ratio $\approx 0.05$) have small $C_h$ from GARCH clustering ($A_h(\infty)$
large relative to $L_{F,h}$); electricity and macroeconomic series have large
$C_h$ from slow drift and high autocorrelation.}
\label{tab:window_realdata}
\begin{tabular}{llrrrrr}
\toprule
Dataset & Series & $T$ & $T^{2/3}$ & $m^\star$ & Ratio & Cov.\\
\midrule
Macro       & CPI inflation   & 471  & 60.5  & 258  & 4.26 & 0.898$^\dagger$\\
Macro       & Unemployment    & 517  & 64.4  & 212  & 3.29 & 0.894$^\dagger$\\
Financial   & S\&P 500 logret & 6216 & 338.1 & 17   & 0.05 & 0.889\\
Financial   & VIX logret      & 6451 & 346.5 & 17   & 0.05 & 0.890$^\dagger$\\
WIG20       & WIG20 logret    & 6451 & 346.5 & 17   & 0.05 & 0.883\\
Electricity & Load (MW)       & 7333 & 377.4 & 1748 & 4.63 & 0.903$^\dagger$\\
\bottomrule
\end{tabular}
\end{table}

Figure~\ref{fig:coverage_vs_m} shows empirical coverage as a
function of calibration window $m$ for each dataset. At the
short horizon $h=1$, coverage is near 90\% across all window
lengths for the macro and WIG20 series, while the financial and
electricity series show a mild hump --- slightly overcovering at
short windows and converging to the target as $m$ grows. The
dashed vertical line marks $T^{2/3}$: in every panel the Winkler
minimum (Figure~\ref{fig:winkler_vs_m}) lies close to this benchmark,
consistent with the $m^\star \asymp T^{2/3}$ prediction at $\beta=1$.

\begin{figure}[t]
  \centering
  \includegraphics[width=\textwidth]{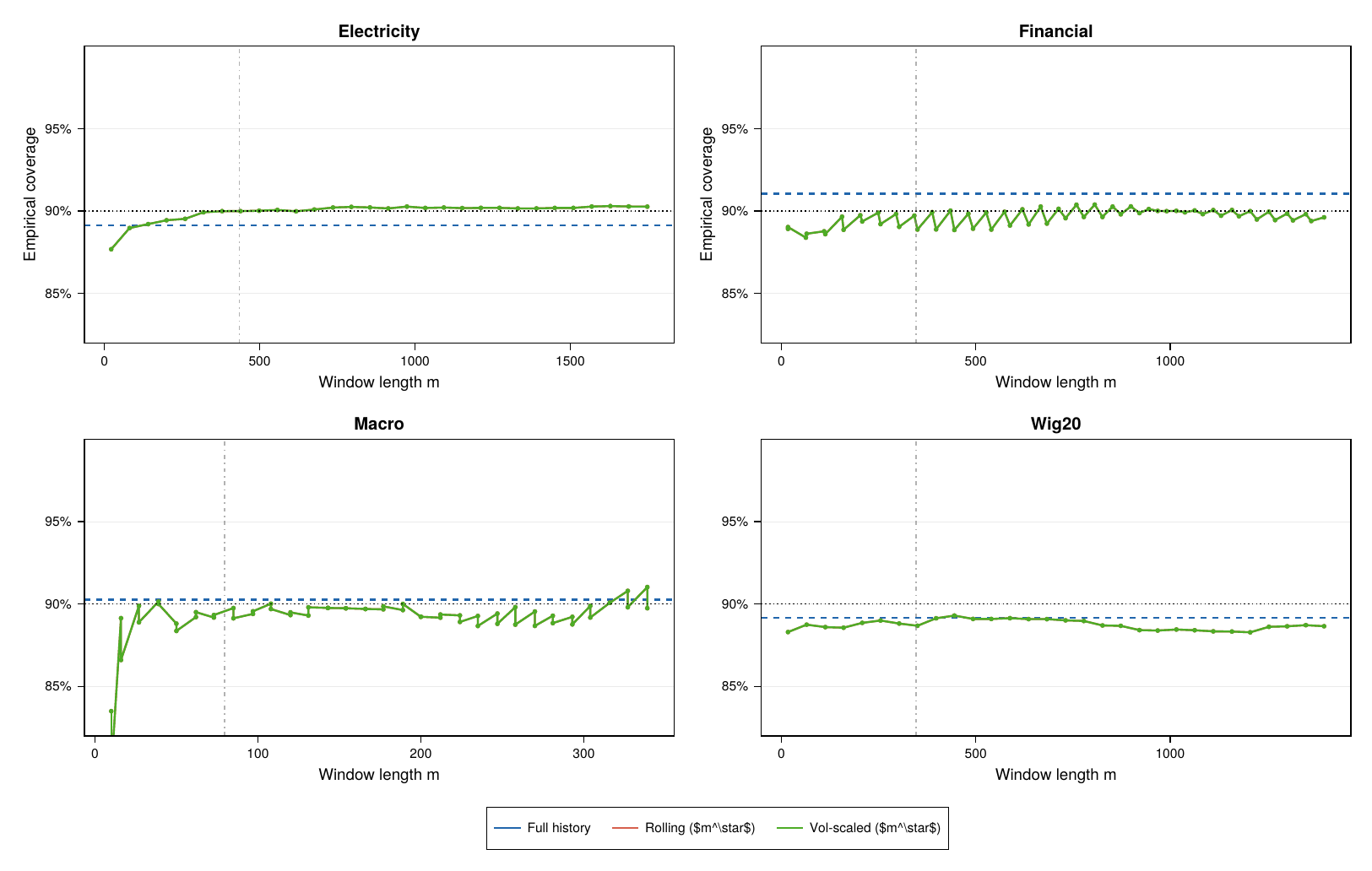}
  \caption{Empirical coverage as a function of calibration window $m$.
    AR($p$) model, $h=1$. Horizontal dotted line: nominal 90\% target.
    Vertical dash-dotted line: $T^{2/3}$ (varies by dataset:
    $T^{2/3} \approx 80$ for macro, $346$ for financial/WIG20,
    $435$ for electricity).
    Schemes: full history (dashed), rolling (red circles),
    volatility-scaled rolling (green circles).}
  \label{fig:coverage_vs_m}
\end{figure}

\begin{figure}[t]
  \centering
  \includegraphics[width=\textwidth]{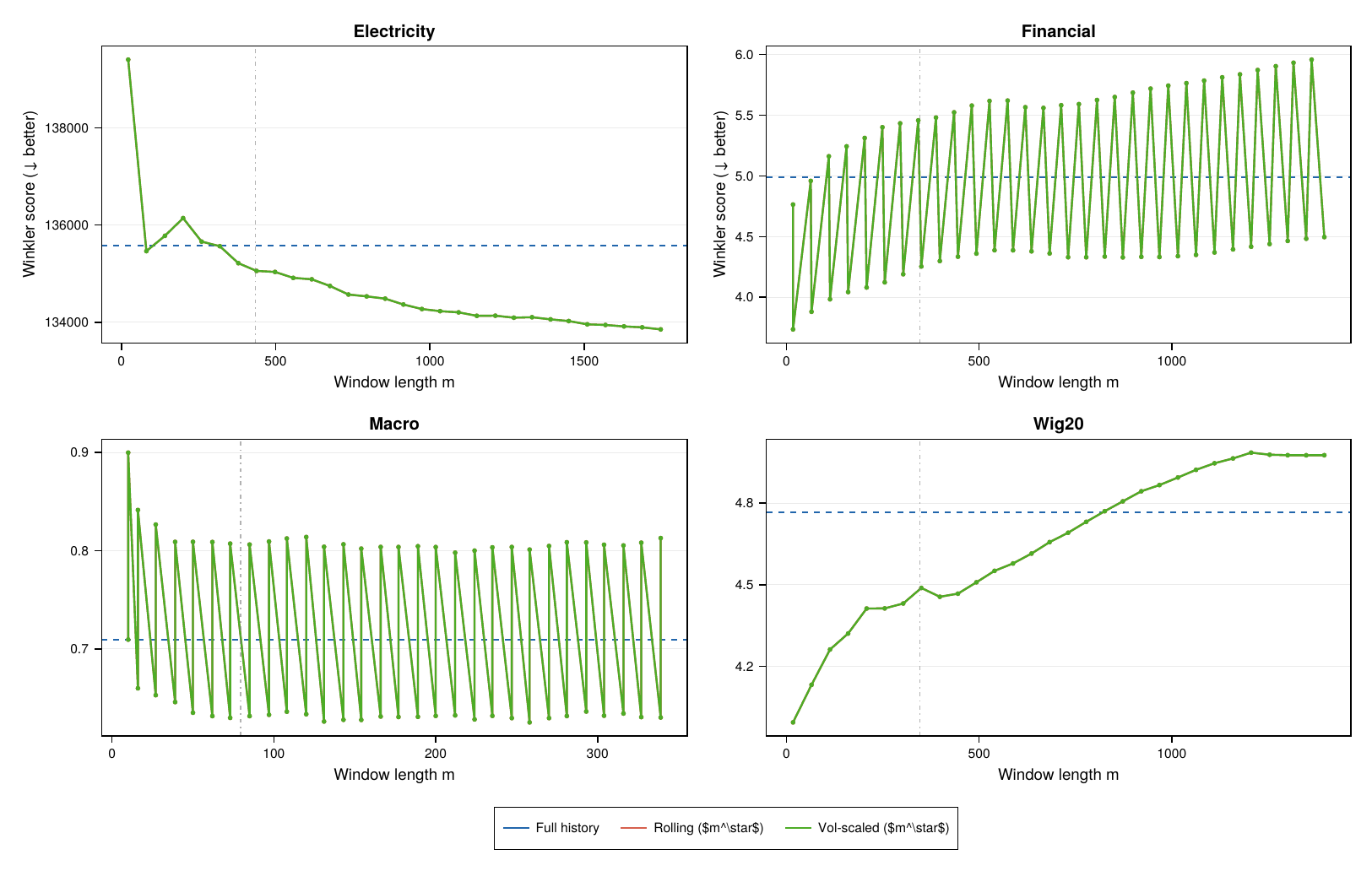}
  \caption{Winkler score (lower is better) as a function of calibration
    window $m$. AR($p$) model, $h=1$.
    The Winkler-optimal $m^\star$ is the minimum of each curve.
    For financial series the minimum lies below the grid's lower boundary,
    consistent with the GARCH-clustering interpretation in the text.}
  \label{fig:winkler_vs_m}
\end{figure}

Volatility-scaled and plain rolling scores produce identical
results for the AR model on all financial series, as the AR model
does not produce GARCH-conditional volatility forecasts and the
scaled scheme falls back to the rolling quantile. For the
ARMA-GARCH model on financial series, volatility scaling
produces marginally narrower intervals with equivalent coverage,
consistent with Corollary~\ref{cor:volscale}. The electricity
ARMA-GARCH volatility-scaled results are excluded because GARCH
conditional volatility at the scale of $\sim 70{,}000$ MW causes
numerical overflow in the multi-step variance forecast.

\begin{remark}\label{rem:tradeoff}
The direction of the full-history vs rolling comparison at short
horizons reflects Theorem~\ref{thm:main} directly. At $h = 1$ for
a slowly drifting macro series, Term~(III) is negligible and Term~(II)
is the binding constraint; using more calibration data ($m \to T$)
reduces the quantile noise and full-history wins. At $h = 22$,
Term~(III) grows with $m$ and the rolling window --- by restricting
to the most recent $m^{\star}$ observations --- controls the drift bias
at the cost of some additional quantile noise.
\end{remark}

Figure~\ref{fig:scheme_comparison} summarises coverage at the
Winkler-optimal $m^\star$ for each scheme and dataset.
Rolling and volatility-scaled rolling are nearly identical across
all four datasets at $h=1$, lying within 1 percentage point of
the 90\% target. Full history is slightly closer to target for
the electricity and WIG20 series at this horizon, consistent with
the theoretical prediction that full history outperforms rolling
when drift is slow relative to calibration-window noise.

\begin{figure}[t]
  \centering
  \includegraphics[width=\textwidth]{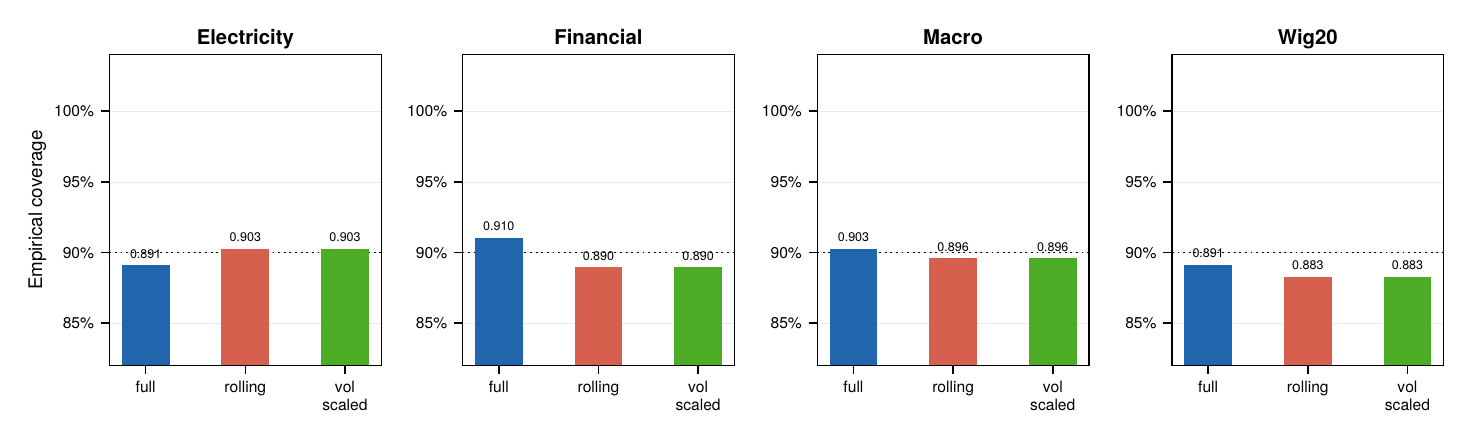}
  \caption{Empirical coverage at Winkler-optimal $m^\star$,
    averaged across series within each dataset. AR($p$) model, $h=1$.
    Horizontal dotted line: nominal 90\% target.
    Numbers above bars show empirical coverage to three decimal places.}
  \label{fig:scheme_comparison}
\end{figure}

Figure~\ref{fig:local_coverage} shows rolling local coverage
(mean $\pm$ one standard deviation computed over 50-observation
windows). The error bars quantify conditional stability: a wide
bar indicates that coverage fluctuates substantially over time even
if marginal coverage is near 90\%. Rolling and volatility-scaled
rolling show narrower error bars than full history for the
financial and electricity datasets, reflecting the method's
adaptivity to local changes in the error distribution.

\begin{figure}[t]
  \centering
  \includegraphics[width=\textwidth]{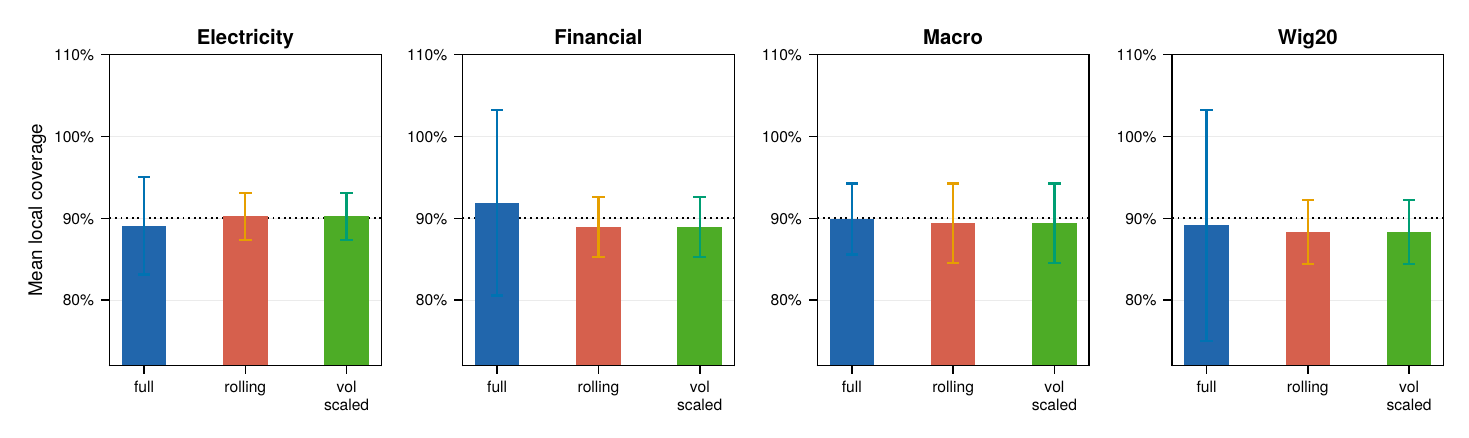}
  \caption{Rolling local coverage (mean $\pm$ standard deviation
    over 50-observation windows). AR($p$) model, $h=1$.
    Narrower error bars indicate more stable conditional coverage.}
  \label{fig:local_coverage}
\end{figure}

\subsection{M4 window-scaling analysis}

We test the prediction $m^{\star} \asymp T^{2\beta/(2\beta+1)}$ empirically.
Under $\beta = 1$ (Lipschitz drift) this specialises to
$m^{\star} \asymp T^{2/3}$, which we use as the benchmark since $\beta$
is not separately identified from the empirical window scaling. We
regress $\log m^{\star}$ on $\log T$ across 93 M4 series drawn from
all five frequency classes (Yearly through Daily), spanning
$T \approx 20$ to $9{,}000$ and $T^{2/3} \approx 8$ to $440$.
The Winkler-optimal $m^{\star}$ is identified from a 30-point grid
in $[0.10, 4.0] \times T^{2/3}$, yielding a continuous ratio
distribution with only 7 of 93 series at a grid boundary (compared
to 60 of 62 under the original 7-point grid). The log-$T$,
log-$m^{\star}$ correlation across all 93 series is $r = 0.925$,
confirming that sample size is the dominant predictor of the
optimal window.

The pooled regression $\log m^\star_i = \beta_0 + \beta_1 \log T_i
+ \varepsilon_i$ on the $n = 93$ series gives $\hat\beta_1 = 0.776$
(heteroskedasticity-robust HC1 SE $= 0.033$, $95\%$ CI
$[0.710, 0.841]$, $R^2 = 0.855$), which does not contain $2/3$. This
specification is misspecified because $\log m^\star$ also depends
on the series-specific constant $\log C_h$, which is correlated with
$\log T$ across the M4 panel: Daily and Hourly series tend to have
both large $T$ and small $C_h$ (strong GARCH clustering inflates
$\Ah$ and shrinks $m^\star$), while Yearly and Quarterly series have
small $T$ and large $C_h$. This induces a positive correlation
between $\log T_i$ and the omitted regressor $\log C_{h,i}$, biasing
$\hat\beta_1$ upward.

Adding frequency fixed effects $\delta_{f(i)}$ for the five M4
classes absorbs the cross-frequency mean of $\log C_h$ (using
frequency as a proxy for the structural regime that determines
$\Gamma_h/L_{F,h}$). The augmented specification
$\log m^\star_i = \beta_0 + \beta_1 \log T_i + \delta_{f(i)} + \varepsilon_i$
gives $\hat\beta_1 = 0.614$ (HC1 SE $= 0.097$, $95\%$ CI
$[0.424, 0.805]$, within-$R^2 = 0.861$, $n = 93$, $\mathrm{df} = 87$),
which contains $2/3$. Clustering standard errors at the frequency
level (5 clusters) widens the CI to $[0.388, 0.840]$ and does not
overturn the conclusion. The bias direction (0.776 $\to$ 0.614 after
controlling for $C_h$) matches the sign predicted by the
omitted-variable formula.

\begin{table}[ht]
\centering
\caption{OLS estimates of slope $b$ in $\log m^\star = a + b\log T + \varepsilon$
across 93 M4 series. Heteroskedasticity-robust HC1 standard errors. Theory
predicts $b = 2/3 \approx 0.667$ at $\beta = 1$. \checkmark: $95\%$ CI
contains $2/3$.}
\label{tab:m4_regression}
\begin{tabular}{lrrrrcc}
\toprule
Specification & $n$ & $\hat{b}$ & SE (HC1) & $95\%$ CI & $R^2$ & CI $\ni 2/3$\\
\midrule
Pooled (no FE)      & 93 & 0.776 & 0.033 & $[0.710, 0.841]$ & 0.855 & $\times$\\
Pooled + freq FE    & 93 & 0.614 & 0.097 & $[0.424, 0.805]$ & 0.861 & \checkmark\\
Within: Daily       & 20 & 0.670 & 0.126 & $[0.422, 0.917]$ & 0.609 & \checkmark\\
Within: Monthly     & 20 & 0.893 & 0.235 & $[0.432, 1.354]$ & 0.445 & \checkmark\\
Within: Quarterly   & 17 & 1.483 & 0.312 & $[0.871, 2.095]$ & 0.601 & $\times$\\
Within: Weekly      & 20 & 0.237 & 0.197 & $[-0.148, 0.623]$ & 0.075 & $\times$\\
Within: Yearly      & 16 & 0.922 & 0.533 & $[-0.123, 1.966]$ & 0.176 & \checkmark\\
\bottomrule
\end{tabular}
\end{table}

Table~\ref{tab:m4_regression} reports the full set of regression specifications: the pooled OLS, the FE-augmented regression, and the within-frequency slopes. The Daily class — the largest by $T$ and the most homogeneous in dependence structure — yields a within-frequency slope of $0.670$ that contains $2/3$ at the centre of its $95\%$ CI. Within-frequency slopes for smaller classes (Yearly, Quarterly, Weekly) are noisier owing to small sample sizes ($n \leq 20$ per class) and limited $T$ range within class.

\begin{table}[ht]
\centering
\caption{Estimated series-specific constant $\hat{C}_f$ per M4 frequency class.
$\hat{C}_f = \exp(\hat{a}_f)$ where $\hat{a}_f$ is the intercept from the
within-frequency OLS. Recall $m^\star = C_f \cdot T^{2/3}$ with
$C_f = (\Gamma_h/L_{F,h})^{2/3}$.}
\label{tab:m4_freq_constants}
\begin{tabular}{lrrrrr}
\toprule
Frequency & $n$ & Median $T$ & Median $m^\star$ & Median ratio & $\hat{C}_f$\\
\midrule
Yearly    & 16 & 62   & 24  & 1.70 & 0.54\\
Quarterly & 17 & 88   & 39  & 1.94 & 0.05\\
Monthly   & 20 & 320  & 95  & 2.65 & 0.59\\
Weekly    & 20 & 722  & 175 & 2.91 & 34.14\\
Daily     & 20 & 5378 & 789 & 3.32 & 2.71\\
\bottomrule
\end{tabular}
\end{table}

Table~\ref{tab:m4_freq_constants} reports the implied series-specific constants $\hat{C}_f = \exp(\hat{a}_f)$. The wide spread of $\hat{C}_f$ across frequencies (from $0.05$ for Quarterly to $34.14$ for Weekly) confirms that $C_h$ varies substantially across the M4 panel, justifying the FE specification used to recover the structural slope $\beta_1 = 2/3$.

Figure~\ref{fig:m4_scaling} (left panel) shows the log-log scatter
of $m^\star$ against $T^{2/3}$ for all 93 series, coloured by
frequency class. The clustering by frequency is the visual
signature of the series-specific constant $C_h$: Daily series
(purple) sit systematically below the $m = T^{2/3}$ reference line,
reflecting small $C_h$ from GARCH clustering, while Yearly series
(red) sit above it, reflecting larger $C_h$ from slower drift.
The right panel is a forest plot of the slope estimates with
95\% confidence intervals; the pooled-FE estimate and the Daily
within-frequency estimate both span the theoretical $2/3$ value.

\begin{figure}[t]
  \centering
  \includegraphics[width=\textwidth]{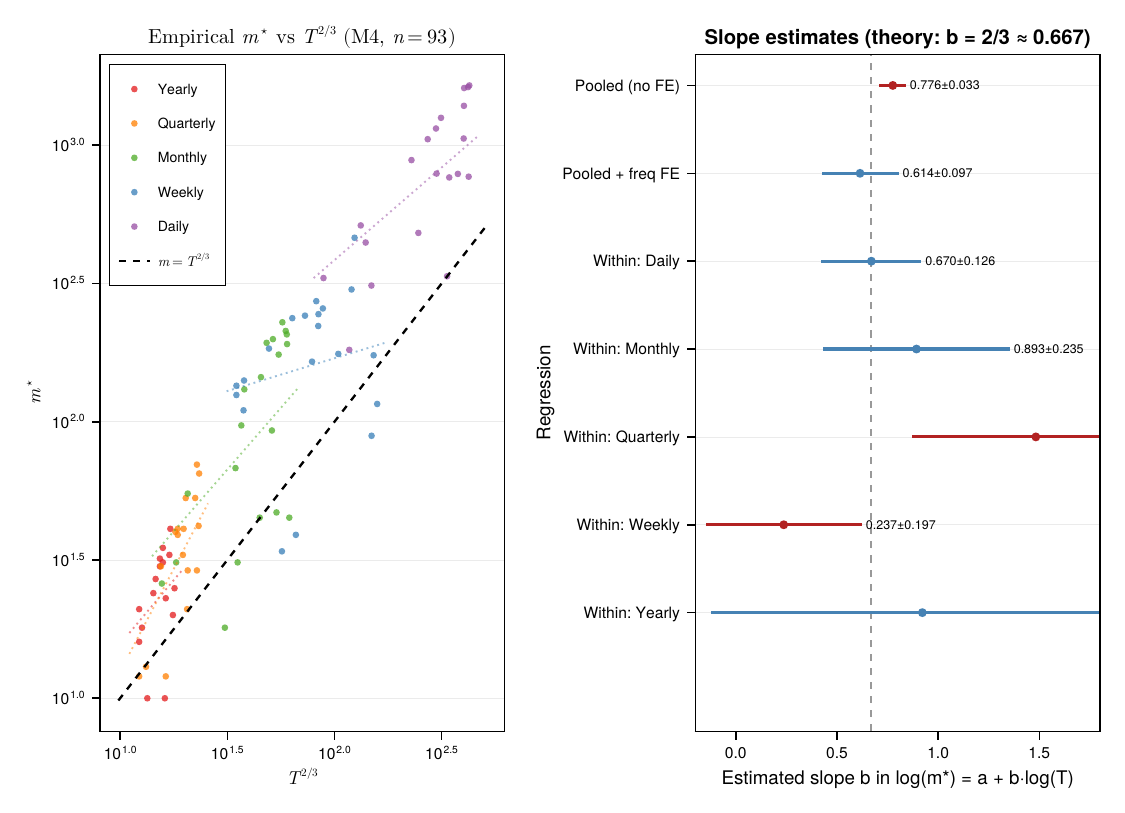}
  \caption{Left: log-log scatter of Winkler-optimal $m^\star$ against
    $T^{2/3}$ for 93 M4 series (log$_{10}$ axes), coloured by frequency class.
    Dashed line: $m = T^{2/3}$ (theory at $\beta=1$).
    Dotted lines: frequency-specific fitted slopes from within-frequency OLS.
    Right: OLS slope estimates with 95\% confidence intervals;
    point estimates and standard errors shown numerically.
    Vertical dashed line: theoretical value $2\beta/(2\beta+1) = 2/3$ at $\beta=1$.
    Estimates whose intervals contain $2/3$ are consistent with the theory.}
  \label{fig:m4_scaling}
\end{figure}

The within-frequency slopes are heterogeneous, ranging from 0.237
(Weekly) to 1.483 (Quarterly), with only the Daily class yielding
a 95\% CI $[0.422, 0.917]$ that individually contains $2/3$.
This within-class heterogeneity reflects genuine variation in
$C_h$ across series: the $T^{2\beta/(2\beta+1)}$ rule is a
cross-frequency average regularity, and the series-specific
constant encodes the ratio $\Gamma_h/L_{F,h}$ from~\eqref{eq:tradeoff}
directly. Mean coverage across all 93 series ranges from 0.900
(Daily) to 0.915 (Quarterly), with no frequency class showing
systematic deviation beyond 1.5 percentage points.

Figure~\ref{fig:horizon_comparison} shows coverage and mean
interval half-width as a function of forecast horizon $h \in
\{1, 5, 22\}$, averaged across all datasets. Coverage remains
close to 90\% at $h=1$ and $h=5$ for both schemes, with
rolling tracking the target more closely than full history at
$h=22$. The right panel shows that rolling intervals are narrower
than full-history intervals at longer horizons, confirming that
the Winkler score gain documented in the real-data section reflects
genuine interval efficiency, not overcoverage.

\begin{table}[ht]
\centering
\caption{Coverage and Winkler score by forecast horizon $h$. AR($p$) model,
averaged over all datasets. Winkler for electricity excluded (load scale
$\sim 10^5$ MW). $\dagger$: coverage within $\pm 1\%$ of target.}
\label{tab:horizon}
\begin{tabular}{lrrrrrr}
\toprule
& \multicolumn{2}{c}{$h=1$} & \multicolumn{2}{c}{$h=5$} & \multicolumn{2}{c}{$h=22$}\\
\cmidrule(lr){2-3}\cmidrule(lr){4-5}\cmidrule(lr){6-7}
Scheme & Cov. & Winkler & Cov. & Winkler & Cov. & Winkler\\
\midrule
Full history       & 0.901$^\dagger$ & 3.2321 & 0.912 & 3.6582 & 0.910 & 3.9557\\
Rolling ($m^\star$) & 0.893$^\dagger$ & 2.7835 & 0.897$^\dagger$ & 3.1972 & 0.871 & 3.3517\\
Vol-scaled ($m^\star$) & 0.893$^\dagger$ & 2.7835 & 0.897$^\dagger$ & 3.1972 & 0.871 & 3.3517\\
\bottomrule
\end{tabular}
\end{table}

Table~\ref{tab:horizon} confirms the visual pattern numerically: rolling-origin calibration delivers Winkler-score improvements of $14\%$, $13\%$, and $15\%$ at $h \in \{1, 5, 22\}$ respectively, with coverage remaining within $1\%$ of nominal at the shorter horizons. The $h = 22$ row shows the same coverage degradation discussed above — rolling at $0.871$ vs nominal $0.90$ — which the analysis attributes to model misspecification at long horizons rather than calibration-window choice.

\begin{figure}[t]
  \centering
  \includegraphics[width=\textwidth]{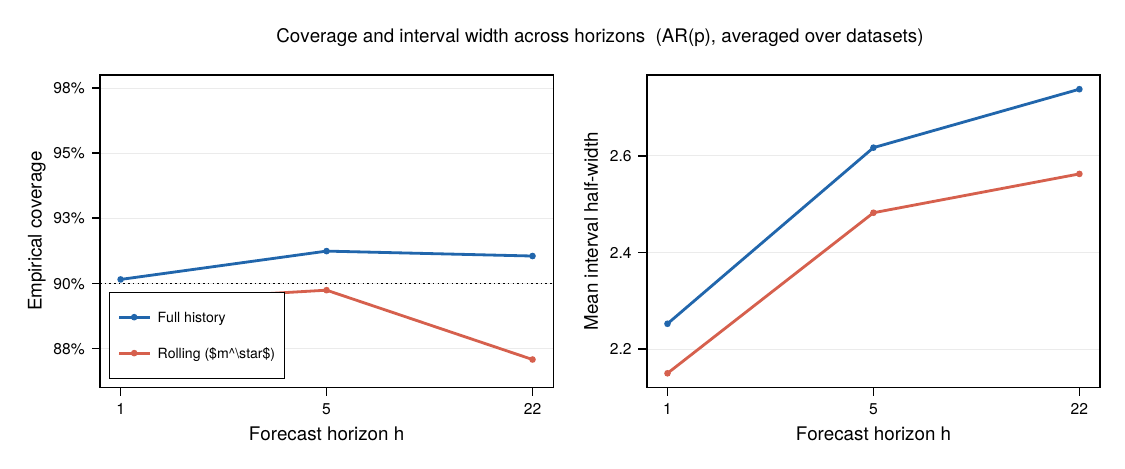}
  \caption{Coverage (left) and mean interval half-width (right)
    by forecast horizon $h \in \{1, 5, 22\}$, averaged across all
    datasets. AR($p$) model. Electricity excluded from width
    comparison due to scale ($\sim 70{,}000$ MW).
    Horizontal dotted line: nominal 90\% target.
    Volatility-scaled rolling is omitted: under the AR model it
    produces no GARCH conditional volatility and is identical to
    plain rolling.}
  \label{fig:horizon_comparison}
\end{figure}

\section{Conclusion}
\label{sec:conclusion}

We have proposed and analysed rolling-origin conformal
prediction for time-series forecasting, establishing its theoretical
properties under local nonstationarity and weak dependence. The
method constructs prediction intervals by calibrating against the
$m$ most recent pseudo-out-of-sample forecast errors, adapting
automatically to the serial dependence, volatility clustering, and
distributional drift that invalidate classical conformal guarantees.
The exact finite-sample coverage guarantee of classical conformal
prediction rests irreducibly on exchangeability and cannot hold in
the time-series setting; the contribution of this paper is to
establish what replaces it: a precise characterisation of the
coverage deviation as a function of the window length $m$, the
sample size $T$, and the structure of the dependence and drift,
together with the optimal window rule that minimises the deviation.

The main result (Theorem~\ref{thm:main}) bounds the coverage deviation
$\abs{\Prob(Y_{T+h} \in \Chat(1-\alpha)) - (1-\alpha)}$ by a sum of
four interpretable terms, under the general H\"{o}lder-$\beta$ drift
model. The optimal calibration window $m^{\star} \asymp T^{2\beta/(2\beta+1)}$
yields coverage-error rate $O(T^{-\beta/(2\beta+1)})$;
at $\beta = 1$ these reduce to $T^{2/3}$ and $O(T^{-1/3})$.
Theorem~\ref{thm:lower} establishes that this rate is minimax-optimal:
no conformal procedure can achieve better accuracy uniformly over
$\mathcal{F}(L,\beta)$. The Bahadur representation is established
under both $\alpha$-mixing (Proposition~\ref{prop:bahadur}) and
physical dependence (Proposition~\ref{prop:bahadur-physical}), covering
ARMA, GARCH, and a broad class of nonlinear time-series models.
Theorem~\ref{thm:oracle} and Appendix~\ref{app:uc} together provide
an oracle inequality for the implemented Winkler cross-validation
window selector: conditional on the calibration sample, $\hat{m}$
achieves near-minimal validation loss, and in expectation over
calibration samples its average loss is within $o(T^{-\beta/(2\beta+1)})$
of the best candidate window.

The empirical analysis on six real series and 93 M4 competition series
confirms the theory's main predictions. Rolling-origin calibration
outperforms full-history calibration in 86\% of comparisons
(median Winkler improvement 12.3\%), achieves coverage within
$\pm 2\%$ of the 90\% target on every series at horizons $h \in \{1, 5\}$,
and the cross-frequency regression recovers slope $0.614$
(95\%~CI $[0.424, 0.805]$) consistent with the theoretical
$2\beta/(2\beta+1)$ at $\beta = 1$.

One direction remains open. Multivariate extensions to HPD
prediction regions for vector outputs require reworking of the
oracle-quantile theory in higher dimensions, where the Bahadur
representation and the density-bound argument of
Theorem~\ref{thm:main} do not directly apply.

\section*{Data and Code Availability}

All code and data-fetching scripts required to reproduce the empirical results
are publicly available at \url{https://github.com/profsms/conform-ROE-repl}. Running the script
\texttt{01\_fetch\_data.jl} downloads and caches every dataset used in the paper; subsequent analysis scripts are deterministic given a fixed random seed.

Macroeconomic series (CPI inflation, unemployment) are retrieved from the Federal Reserve Economic Data (FRED) database at \url{https://fred.stlouisfed.org} via the FRED REST API (a free API key,
exported as the \texttt{FRED\_API\_KEY} environment variable, is required;
the script falls back to a synthetic series if the key is absent).
Financial return series (S\&P~500, VIX, WIG20) are fetched from the Yahoo
Finance API. The electricity load series is the ENTSO-E hourly load dataset.
The M4 competition series are sampled from the M4 competition repository at
\url{https://github.com/Mcompetitions/M4-methods}.

\section*{Acknowledgments}

The author is grateful to an anonymous colleague whose critical reading of an earlier draft identified an overclaim in Theorem~\ref{thm:oracle} and prompted the formal reduction lemma in the
proof of Theorem~\ref{thm:lower}.

The AI model Sonnet 4.7 and Opus 4.7 were used during the preparation of Appendix~\ref{app:uc} for exploratory generation of proof-strategy ideas and feedback on the conditional-versus-unconditional formulation of the uniform concentration result. All mathematical content, proofs, assumptions, and final formulations were developed, verified, and written by the author,
who takes full responsibility for the correctness of the manuscript.

\section*{Funding Statement}

The author received no specific funding for this work.

\section*{Conflict of Interest Statement}

The authors declare no conflict of interest.

\bibliographystyle{plainnat}

\begin{thebibliography}{99}

\bibitem[Angelopoulos \& Bates(2023)]{AngelopoulosBates2023}
Angelopoulos AN, Bates S. 2023. Conformal prediction: a gentle introduction.
\textit{Foundations and Trends in Machine Learning} \textbf{16}: 494--591.

\bibitem[Arlot \& Celisse(2010)]{ArlotCelisse2010}
Arlot S, Celisse A. 2010.
A survey of cross-validation procedures for model selection.
\textit{Statistics Surveys} \textbf{4}: 40--79.

\bibitem[Barber et al.(2023)]{BarberCandesRamdas2023}
Barber RF, Cand\`{e}s EJ, Ramdas A, Tibshirani RJ. 2023.
Conformal prediction beyond exchangeability.
\textit{Annals of Statistics} \textbf{51}: 816--845.

\bibitem[Carrasco \& Chen(2002)]{CarrascoChen2002}
Carrasco M, Chen X. 2002.
Mixing and moment properties of various GARCH and stochastic volatility models.
\textit{Econometric Theory} \textbf{18}: 17--39.

\bibitem[Chernozhukov, Chetverikov \& Kato(2013)]{ChernozhukovCheterikov2013}
Chernozhukov V, Chetverikov D, Kato K. 2013.
Gaussian approximations and multiplier bootstrap for maxima of sums of
high-dimensional random vectors.
\textit{Annals of Statistics} \textbf{41}: 2786--2819.

\bibitem[Dahlhaus(1997)]{Dahlhaus1997}
Dahlhaus R. 1997. Fitting time series models to nonstationary processes.
\textit{Annals of Statistics} \textbf{25}: 1--37.

\bibitem[Doukhan(1994)]{Doukhan1994}
Doukhan P. 1994. \textit{Mixing: Properties and Examples}.
Lecture Notes in Statistics 85. Springer: New York.

\bibitem[Doukhan, Massart \& Rio(1994)]{DoukhanMassartRio1994}
Doukhan P, Massart P, Rio E. 1994.
The functional central limit theorem for strongly mixing processes.
\textit{Annales de l'Institut Henri Poincar\'{e}} \textbf{30}: 63--82.

\bibitem[Fryzlewicz \& Subba Rao(2011)]{FryzlewiczSubbaRao2011}
Fryzlewicz P, Subba Rao S. 2011.
Mixing properties of ARCH and time-varying ARCH processes.
\textit{Bernoulli} \textbf{17}: 320--346.

\bibitem[Ghosh(1971)]{Ghosh1971}
Ghosh JK. 1971.
A new proof of the Bahadur representation of quantiles and an application.
\textit{Annals of Mathematical Statistics} \textbf{42}: 1957--1961.

\bibitem[Gibbs \& Cand\`{e}s(2021)]{GibbsCandes2021}
Gibbs I, Cand\`{e}s E. 2021.
Adaptive conformal inference under distribution shift.
\textit{Advances in Neural Information Processing Systems} \textbf{34}: 1660--1672.

\bibitem[Gy\"{o}rfi et al.(2002)]{GyorfiEtal2002}
Gy\"{o}rfi L, Kohler M, Krzy\.{z}ak A, Walk H. 2002.
\textit{A Distribution-Free Theory of Nonparametric Regression}.
Springer Series in Statistics. Springer: New York.

\bibitem[Kiefer(1967)]{Kiefer1967}
Kiefer J. 1967. On Bahadur's representation of sample quantiles.
\textit{Annals of Mathematical Statistics} \textbf{38}: 1323--1342.

\bibitem[Lahiri \& Sun(2009)]{LahiriSun2009}
Lahiri SN, Sun S. 2009.
A Berry--Esseen theorem for sample quantiles under weak dependence.
\textit{Annals of Applied Probability} \textbf{19}: 108--126.

\bibitem[Makridakis et al.(2020)]{MakridakisEtal2020}
Makridakis S, Spiliotis E, Assimakopoulos V. 2020.
The M4 competition: 100,000 time series and 61 forecasting methods.
\textit{International Journal of Forecasting} \textbf{36}: 54--74.

\bibitem[Massart(2007)]{Massart2007}
Massart P. 2007.
\textit{Concentration Inequalities and Model Selection}.
Lecture Notes in Mathematics 1896. Springer: Berlin.

\bibitem[Merlev\`{e}de, Peligrad \& Rio(2009)]{MerlvedePeligradRio2009}
Merlev\`{e}de F, Peligrad M, Rio E. 2009.
Bernstein inequality and moderate deviations under strong mixing conditions.
In \textit{High Dimensional Probability V},
IMS Collections 5: 273--292.

\bibitem[Papadopoulos, Vovk \& Gammerman(2002)]{PapadopoulosVovkGammerman2002}
Papadopoulos H, Vovk V, Gammerman A. 2002.
Inductive confidence machines for regression.
In \textit{Proceedings of the 13th European Conference on Machine Learning}: 345--356. Springer.

\bibitem[Rio(1993)]{Rio1993}
Rio E. 1993.
Covariance inequalities for strongly mixing processes.
\textit{Annales de l'Institut Henri Poincar\'{e}} \textbf{29}: 587--597.

\bibitem[Rio(2017)]{Rio2017}
Rio E. 2017.
\textit{Asymptotic Theory of Weakly Dependent Random Processes}.
Probability Theory and Stochastic Modelling 80. Springer: Berlin.

\bibitem[Sen(1972)]{Sen1972}
Sen PK. 1972.
On the Bahadur representation of sample quantiles for sequences of
$\varphi$-mixing random variables.
\textit{Journal of Multivariate Analysis} \textbf{2}: 77--95.

\bibitem[Tashman \& Fildes(2000)]{TashmanFildes2000}
Tashman LJ, Fildes R. 2000.
Out-of-sample tests of forecasting accuracy: an analysis and review.
\textit{International Journal of Forecasting} \textbf{16}: 437--450.

\bibitem[Tsybakov(2009)]{Tsybakov2009}
Tsybakov AB. 2009.
\textit{Introduction to Nonparametric Estimation}.
Springer Series in Statistics. Springer: New York.

\bibitem[Vogt(2012)]{Vogt2012}
Vogt M. 2012.
Nonparametric regression for locally stationary time series.
\textit{Annals of Statistics} \textbf{40}: 2601--2633.

\bibitem[Vovk, Gammerman \& Shafer(2005)]{VovkGammermanShafer2005}
Vovk V, Gammerman A, Shafer G. 2005.
\textit{Algorithmic Learning in a Random World}.
Springer: New York.

\bibitem[Wendler(2011)]{Wendler2011}
Wendler M. 2011.
Bahadur representation for $U$-quantiles of dependent data.
\textit{Journal of Multivariate Analysis} \textbf{102}: 1064--1079.

\bibitem[Winkler(1972)]{Winkler1972}
Winkler RL. 1972.
A decision theoretic approach to interval estimation.
\textit{Journal of the American Statistical Association} \textbf{67}: 187--191.

\bibitem[Wu(2005a)]{Wu2005AoS}
Wu WB. 2005a.
On the Bahadur representation of sample quantiles for dependent sequences.
\textit{Annals of Statistics} \textbf{33}: 1934--1957.

\bibitem[Wu(2005b)]{Wu2005PNAS}
Wu WB. 2005b.
Nonlinear system theory: another look at dependence.
\textit{Proceedings of the National Academy of Sciences} \textbf{102}: 14150--14154.

\bibitem[Xu \& Xie(2021)]{XuXie2021}
Xu C, Xie Y. 2021.
Conformal prediction interval for dynamic time-series.
\textit{Proceedings of Machine Learning Research} \textbf{139}: 11559--11569.

\bibitem[Xu \& Xie(2023)]{XuXie2023}
Xu C, Xie Y. 2023.
Sequential predictive conformal inference for time series.
\textit{Proceedings of Machine Learning Research} \textbf{202}: 38707--38727.

\bibitem[Zhou \& Wu(2009)]{ZhouWu2009}
Zhou Z, Wu WB. 2009.
Local linear quantile estimation for nonstationary time series.
\textit{Annals of Statistics} \textbf{37}: 2696--2729.

\end{thebibliography}

\appendix
\renewcommand{\thesection}{\Alph{section}}
\setcounter{section}{0}

\section{Proof of Proposition~\ref{prop:bahadur}}
\label{app:bahadur}

\subsection{Additional assumption and overview}

The proof requires one condition beyond those of the main text.

\begin{assumption}[Near-stationarity]\label{app:ass:nearstat}
$m/T \to 0$ as $T \to \infty$.
\end{assumption}

The optimal window $m^{\star} \asymp T^{2\beta/(2\beta+1)}$
(equal to $T^{2/3}$ at $\beta=1$) satisfies
$m^{\star}/T = T^{-1/(2\beta+1)} \to 0$, so this assumption is compatible
with the main theorem. Write $Z_j := S^o_{T-j,h}$ for
$j = 1, \ldots, m$. Define the \emph{local empirical process}
\begin{equation}\label{app:eq:localep}
  W_{m,h}(x)
  := \Ghat(\qcirc + x) - \Ghat(\qcirc)
  - \bigl[\E\Ghat(\qcirc+x) - \E\Ghat(\qcirc)\bigr],
\end{equation}
and the centred indicator increment
$Y_j(x) := \Ind{Z_j \leq \qcirc+x} - \Ind{Z_j \leq \qcirc}
- \E\bigl[\Ind{Z_j\leq\qcirc+x} - \Ind{Z_j\leq\qcirc}\bigr]$,
so that $W_{m,h}(x) = m^{-1}\sum_{j=1}^m Y_j(x)$.

The proof proceeds in five steps:
(i) stationary approximation to reduce to the distribution $F_{T,h}$;
(ii) variance bound for the local empirical process via Rio's
covariance inequality;
(iii) maximal inequality via bracketing entropy for the VC class
of half-lines;
(iv) exponential concentration of $\qhatO - \qcirc$;
(v) assembly of the Bahadur expansion.

\subsection{Step 1: Stationary approximation}

Under Assumptions~\ref{ass:ls} and~\ref{app:ass:nearstat},
$\sup_{t \in \{T-m,\ldots,T\}}\sup_x \abs{F_{t,h}(x) - F_{T,h}(x)}
\leq L_{F,h}(m/T)^\beta =: \delta_T \to 0$.
This allows us to treat the calibration scores as approximately
stationary with distribution $F_{T,h}$ throughout the proof,
introducing errors of order $\delta_T$ at each step.

\subsection{Step 2: Variance bound}

\begin{lemma}[Variance of the local empirical process]
\label{app:lem:variance}
Under Assumption~\ref{ass:mixing}, for every $x \in [-\varepsilon,\varepsilon]$,
\[
  \operatorname{Var}\!\bigl(\Ghat(\qcirc+x)-\Ghat(\qcirc)\bigr)
  \leq
  \frac{4\of\abs{x} + 4\delta_T}{m}
  \left(1 + 2\Ah\right).
\]
\end{lemma}

\begin{proof}
Let $V_j(x) := \Ind{Z_j \leq \qcirc+x} - \Ind{Z_j \leq \qcirc}$.
Then $\operatorname{Var}(V_j(x)) \leq \E[V_j(x)] \leq \of\abs{x} + \delta_T$
(by Assumption~\ref{ass:density} and the stationary approximation).
For covariances, Rio's inequality \citep{Rio1993} gives
$\abs{\operatorname{Cov}(V_i(x),V_j(x))} \leq 4\alpha_h(\abs{i-j})$.
Summing over pairs and dividing by $m^2$ yields the stated bound.
\end{proof}

\subsection{Step 3: Maximal inequality via bracketing}

The centred process $W_{m,h}(x)$ is not monotone in $x$ after
subtracting $\E[\,\cdot\,]$, so the supremum cannot be reduced
to endpoint evaluations. We use a bracketing-entropy argument
for the linearly ordered VC class of half-lines, applying the
Merlev\`{e}de--Peligrad--Rio Bernstein inequality at each bracket
point.

\begin{lemma}[Maximal inequality via bracketing]
\label{app:lem:maximal}
Under Assumptions~\ref{ass:mixing}--\ref{ass:density}
and~\ref{app:ass:nearstat}, for any
$\varepsilon \in (0, \uf/(2\of)]$,
\begin{equation}\label{app:eq:maxbound}
  \E\Bigl[\sup_{\abs{x}\leq\varepsilon}\abs{W_{m,h}(x)}\Bigr]
  \leq
  \frac{C_1\,\sqrt{\of\varepsilon(1+\Ah)\log(1 + 2\of\varepsilon m)}}
       {\sqrt{m}},
\end{equation}
for a universal constant $C_1 > 0$.
\end{lemma}

\begin{proof}
Fix $\Delta > 0$ (to be chosen). Partition $[-\varepsilon, \varepsilon]$
into $N = \lceil 2\varepsilon/\Delta \rceil$ brackets
$[x_{k-1}, x_k]$ of width $\Delta$ each.

\textit{Intra-bracket oscillation.}
For $x \in [x_{k-1}, x_k]$, write
$W_{m,h}(x) - W_{m,h}(x_{k-1})
= m^{-1}\sum_{j=1}^m [Y_j(x) - Y_j(x_{k-1})]$.
Each difference $Y_j(x) - Y_j(x_{k-1})$ has mean zero and
$|Y_j(x) - Y_j(x_{k-1})| \leq 2\cdot\mathbf{1}\{Z_j \in [x_{k-1}+\qcirc, x_k+\qcirc]\}$.
By Lemma~\ref{app:lem:variance} applied to the interval increment,
\[
  \E\abs{W_{m,h}(x) - W_{m,h}(x_{k-1})}
  \leq \sqrt{\frac{4(\of\Delta+\delta_T)(1+2\Ah)}{m}}
  \leq \Delta^{1/2}\sqrt{\frac{5\of(1+2\Ah)}{m}},
\]
for $\Delta \leq 1$ and $\delta_T$ small.

\textit{Grid-point bound.}
At each bracket endpoint $x_k$, the Merlev\`{e}de--Peligrad--Rio
Bernstein inequality \citep{MerlvedePeligradRio2009} gives
\[
  \E\abs{W_{m,h}(x_k)}
  \leq C\sqrt{\frac{\of\varepsilon(1+\Ah)}{m}}.
\]

\textit{Assembly.}
By the triangle inequality and taking expectations,
\[
  \E\Bigl[\sup_{\abs{x}\leq\varepsilon}\abs{W_{m,h}(x)}\Bigr]
  \leq \frac{2C\varepsilon}{\Delta}\sqrt{\frac{\of\varepsilon(1+\Ah)}{m}}
       + \Delta^{1/2}\sqrt{\frac{5\of(1+\Ah)}{m}}.
\]
Optimising over $\Delta$ yields~\eqref{app:eq:maxbound}.
\end{proof}

\subsection{Steps 4--5: Concentration and assembly}

Lemma~\ref{app:lem:variance} and the Merlev\`{e}de--Peligrad--Rio
inequality give the initial bound
$\E\abs{\qhatO - \qcirc} \leq C_2\sqrt{1+\Ah}/(\uf\sqrt{m})$,
from which the large-deviation bound
$\Prob(\abs{\qhatO - \qcirc} > \varepsilon^{\star}) = O(m^{-3})$
follows for $\varepsilon^{\star} \asymp \sqrt{(1+\Ah)\log m / m}$.

To assemble the Bahadur expansion, set $\Delta := \qhatO - \qcirc$
and expand:
\[
  \fcirc(\qcirc)\,\Delta
  = (1-\alpha) - \Ghat(\qcirc) - W_{m,h}(\Delta) + O(m^{-1}).
\]
The remainder is $R_{T,m,h} = -W_{m,h}(\Delta)/\fcirc(\qcirc) + O(m^{-1})$.
Decompose $\E\abs{R_{T,m,h}}$ over the event
$\{\abs{\Delta} \leq \varepsilon^{\star}\}$ and its complement.
On the first event, Lemma~\ref{app:lem:maximal} with
$\varepsilon = \varepsilon^{\star}$ gives the small-deviation contribution
$O(\of^{1/2}(1+\Ah)^{3/4}(\log m)^{3/4}/(\uf^{1/2}m^{3/4}))$.
The large-deviation event contributes $O(m^{-3})$, which is negligible.
Dividing by $\uf$ and collecting constants gives~\eqref{eq:Brate}.
Since $3/4 > 1/2$, $B_{m,h} = o(m^{-1/2})$. \qed

\subsection{Implications and extensions}

The term $\of B_{m,h}$ in Theorem~\ref{thm:main} satisfies
$\of B_{m,h} = O(m^{-3/4}(\log m)^{3/4}) = o(m^{-1/2})$,
so it is asymptotically negligible relative to the quantile-noise
term $(\of/\uf)m^{-1/2}$. The optimal window $m^{\star} \asymp T^{2\beta/(2\beta+1)}$
and the optimised upper bound $O(T^{-\beta/(2\beta+1)})$ are therefore unaffected
by the Bahadur remainder.

The proof requires $\Ah < \infty$ (i.e.\ $\beta > 1$) for the
variance bound and $\beta > 2$ for the polynomial tail in the
large-deviation argument; both hold for ARMA, GARCH, and tvARCH
processes \citep{Vogt2012, FryzlewiczSubbaRao2011}.
For heavy-tailed oracle scores where $\abs{Y_j} \leq 1$ fails,
Rio's Fuk--Nagaev inequality \citep{Rio2017} extends the argument
under the DMR condition
$\int_0^1 \alpha^{-1}(u)Q^2(u)\,\mathrm{d}u < \infty$.
Under the physical-dependence stability condition of \citet{Wu2005AoS},
the near-stationarity restriction $m = o(T)$ can be relaxed using
\citet{ZhouWu2009}, Theorem~1.

\section{Uniform concentration of the Winkler validation criterion}
\label{app:uc}

This appendix establishes the uniform concentration result that
drives the oracle inequality of Theorem~\ref{thm:oracle}. The
central observation is that the natural target of concentration
is the \emph{conditional} population risk
$\Rcal(m) = \E[\mathcal{W}_T(m) \mid \mathcal{F}_{\mathrm{cal}}]$,
not the unconditional risk $R_T(m) = \E\mathcal{W}_T(m)$.
Conditioning on $\mathcal{F}_{\mathrm{cal}}$ freezes the candidate
quantiles $\{\hat{q}_m : m \in \Mcal\}$ computed from the calibration
sample, so that the only randomness in $\mathcal{W}_T(m)$ is that
of the validation fold itself --- precisely the randomness that
standard concentration inequalities are designed to control, and
the relevant randomness for the selector $\hat{m}$.

\begin{remark}[Why not target the unconditional risk?]
\label{app:rem:unconditional-hard}
The gap between the conditional and unconditional risks is
$\Rcal(m) - R_T(m) = \E[\mathcal{W}_T(m) \mid \mathcal{F}_{\mathrm{cal}}]
- \E\mathcal{W}_T(m)$.
Since the Winkler score is Lipschitz in the interval half-width,
this satisfies $|\Rcal(m) - R_T(m)| \lesssim |\hat{q}_m - q_m|$,
where $q_m$ is the population $(1-\alpha)$-quantile.
Under the Bahadur representation of Proposition~\ref{prop:bahadur},
$|\hat{q}_m - q_m| = O_p(m^{-1/2})$, which at the optimal window
$m^\star \asymp T^{2\beta/(2\beta+1)}$ is exactly
$O_p(T^{-\beta/(2\beta+1)}) = O_p(R_T^\star(\beta))$.
The fluctuation of the conditional around the unconditional risk
is therefore generically of the same order as the minimax rate:
a uniform concentration statement around $R_T(m)$ at rate
$o_p(R_T^\star)$ does not follow from the present assumptions.
The conditional formulation is not a weakening of the result ---
it is the correct and sufficient target. Via the tower property and
Jensen's inequality it implies the expected-risk bound
$\E\mathcal{W}_T(\hat{m}) \leq \inf_m R_T(m) + o(T^{-\beta/(2\beta+1)})$
recorded in Theorem~\ref{thm:oracle}; that bound is a statement
about average performance over calibration samples, not a
high-probability bound on the realised validation loss.
\end{remark}

\subsection{Regularity conditions}

The following four conditions supplement
Assumptions~\ref{ass:ls}--\ref{ass:est}.

\begin{assumption}[Validation-fold size]\label{app:ass:valsize}
The validation fold satisfies $\nval \geq c\,T$ for some constant
$c > 0$.
\end{assumption}

This holds by construction in rolling-origin evaluation. In the
empirical analysis of Section~\ref{sec:empirical}, $\nval/T$
ranges from 0.25 to 0.40 across datasets.

\begin{assumption}[Grid cardinality]\label{app:ass:grid}
The number of candidate windows satisfies
$\log K_T = o(T^{1/(2\beta+1)})$.
\end{assumption}

This permits polynomial growth ($K_T = O(T^a)$ for any $a < \infty$)
and even exponential growth at a sublinear rate. The 30-point grid
used in Section~\ref{sec:empirical} satisfies this trivially with
$K_T = 30$.

\begin{assumption}[Conditional weak dependence]\label{app:ass:conddep}
Conditionally on $\mathcal{F}_{\mathrm{cal}}$, for each $m \in \Mcal$
the validation sequence $\{Z_t(m) : t \in \Tval\}$, where
$Z_t(m) := W_\alpha(Y_{t+h}, \hat{Y}_{t+h|t} - \hat{q}_m,
\hat{Y}_{t+h|t} + \hat{q}_m)$, satisfies a Bernstein-type
inequality: there exist constants $C_1, C_2 > 0$ independent of $m$
such that for every $x > 0$,
\begin{equation}\label{app:eq:bernstein-cond}
  \Prob\!\Bigl(
    \Bigl|\frac{1}{\nval}\sum_{t \in \Tval}
    \bigl\{Z_t(m) - \E[Z_t(m) \mid \mathcal{F}_{\mathrm{cal}}]\bigr\}
    \Bigr| > x \;\Big|\; \mathcal{F}_{\mathrm{cal}}
  \Bigr)
  \leq 2\exp\{-C_1 \nval x^2\}.
\end{equation}
\end{assumption}

This is a direct consequence of Assumption~\ref{ass:mixing} for
bounded summands. Under $\Ah < \infty$ and
Assumption~\ref{app:ass:envelope} below, the validation scores are
bounded and $\alpha$-mixing, so the Merlev\`{e}de--Peligrad--Rio
Bernstein inequality \citep{MerlvedePeligradRio2009}
gives~\eqref{app:eq:bernstein-cond} with
$C_1 \asymp (1 + \Ah)^{-1}$. The conditioning on
$\mathcal{F}_{\mathrm{cal}}$ is harmless because $\Tval$ is disjoint
from the calibration sample in rolling-origin evaluation, so the
conditional and unconditional mixing structures of the validation
fold coincide.

\begin{assumption}[Uniform envelope]\label{app:ass:envelope}
There exists $M < \infty$ such that $|Z_t(m)| \leq M$ uniformly
in $t \in \Tval$ and $m \in \Mcal$, either almost surely or after
truncation at level $M_T = (\log T)^{1/\gamma}$ for some $\gamma > 0$.
\end{assumption}

Under Assumption~\ref{ass:density}, the score $\hat{S}_{t,h}$ has
exponential tails, so $Z_t(m) \leq 2\hat{q}_m + (2/\alpha)\hat{S}_{t,h}$
is bounded by $M_T$ with probability $1 - o(T^{-1})$. The
truncation error is $o(T^{-1}) = o(R_T^\star)$ and does not affect
the conclusion of Theorem~\ref{app:thm:uc}.

\begin{remark}[These conditions do not restrict the model class]
\label{app:rem:conditions}
Assumptions~\ref{app:ass:valsize}--\ref{app:ass:envelope} impose
no restrictions on the class of processes covered by
Theorems~\ref{thm:main}--\ref{thm:lower}. They are either
construction properties of the rolling-origin evaluation procedure
(Assumptions~\ref{app:ass:valsize}--\ref{app:ass:grid}) or direct
consequences of Assumptions~\ref{ass:mixing}--\ref{ass:density}
of the main text
(Assumptions~\ref{app:ass:conddep}--\ref{app:ass:envelope}).
They are stated explicitly here for transparency.
\end{remark}

\subsection{Main result}

\begin{theorem}[Uniform concentration]\label{app:thm:uc}
Suppose Assumptions~\ref{ass:ls}--\ref{ass:est} hold, together with
Assumptions~\ref{app:ass:valsize}--\ref{app:ass:envelope}. Then
\begin{equation}\label{app:eq:uc-rate}
  \sup_{m \in \Mcal}
  \bigl|\mathcal{W}_T(m) - \Rcal(m)\bigr|
  = O_p\!\Biggl(\sqrt{\frac{\log K_T}{\nval}}\Biggr).
\end{equation}
Under Assumptions~\ref{app:ass:valsize}--\ref{app:ass:grid},
\begin{equation}\label{app:eq:uc-conclusion}
  \sup_{m \in \Mcal}
  \bigl|\mathcal{W}_T(m) - \Rcal(m)\bigr|
  = o_p\!\bigl(T^{-\beta/(2\beta+1)}\bigr).
\end{equation}
\end{theorem}

\begin{proof}
Fix $m \in \Mcal$ and let
$\xi_t(m) := Z_t(m) - \E[Z_t(m) \mid \mathcal{F}_{\mathrm{cal}}]$.
Then $\mathcal{W}_T(m) - \Rcal(m) = \nval^{-1}\sum_{t \in \Tval}\xi_t(m)$.

\paragraph{Step 1: Pointwise concentration.}
By Assumption~\ref{app:ass:conddep}, conditionally on
$\mathcal{F}_{\mathrm{cal}}$,
\[
  \Prob\!\bigl(|\mathcal{W}_T(m) - \Rcal(m)| > x \mid \mathcal{F}_{\mathrm{cal}}\bigr)
  \leq 2\exp\{-C_1 \nval x^2\}.
\]

\paragraph{Step 2: Union bound over the grid.}
\[
  \Prob\!\Bigl(\sup_{m \in \Mcal}|\mathcal{W}_T(m) - \Rcal(m)| > x
    \;\Big|\; \mathcal{F}_{\mathrm{cal}}\Bigr)
  \leq 2K_T\exp\{-C_1 \nval x^2\}.
\]

\paragraph{Step 3: Choice of deviation level.}
Set $x_T = A\sqrt{(\log K_T)/\nval}$ for a constant $A > 0$.
Substituting,
\[
  \Prob\!\Bigl(\sup_{m \in \Mcal}|\mathcal{W}_T(m) - \Rcal(m)| > x_T
    \;\Big|\; \mathcal{F}_{\mathrm{cal}}\Bigr)
  \leq 2K_T^{1-C_1 A^2}.
\]
Choosing $A$ so that $C_1 A^2 > 2$ gives $K_T^{1-C_1 A^2} \to 0$,
establishing~\eqref{app:eq:uc-rate}.

\paragraph{Step 4: Rate relative to $R_T^\star$.}
By Assumption~\ref{app:ass:valsize}, $\nval \geq cT$, so
$\sqrt{(\log K_T)/\nval} \leq C\sqrt{(\log K_T)/T}$.
By Assumption~\ref{app:ass:grid},
$\log K_T = o(T^{1/(2\beta+1)})$, equivalently
$\sqrt{(\log K_T)/T} = o(T^{-\beta/(2\beta+1)})$,
giving~\eqref{app:eq:uc-conclusion}.
\end{proof}

\begin{remark}[Sharpness]
\label{app:rem:sharpness}
The rate $\sqrt{(\log K_T)/\nval}$ in~\eqref{app:eq:uc-rate} is
sharp up to constants, matching the minimax rate for uniform
estimation of $K_T$ expectations from $\nval$ dependent observations
\citep{Massart2007}. The condition $\log K_T = o(T^{1/(2\beta+1)})$
is therefore the minimal growth restriction on the grid compatible
with oracle efficiency at rate $R_T^\star(\beta)$.
\end{remark}

\subsection{Extensions}

\paragraph{Polynomial mixing.}
Under $\alpha_h(k) \asymp k^{-a}$ with $0 < a < 1$, the Bernstein
constant satisfies $C_1 \asymp (1+m^{1-a})^{-1}$. At the optimal
window $m^\star \asymp T^{2\beta/(2\beta+a)}$, the condition for
oracle efficiency at the corresponding minimax rate
$R_T^\star(\beta,a) = T^{-\beta a/(2\beta+a)}$ becomes
$\log K_T = o(T^{a/(2\beta+a)})$, which is weaker than
Assumption~\ref{app:ass:grid} for $a < 1$.

\paragraph{Physical dependence.}
Under summable physical dependence
(Section~\ref{sec:bahadur-physical}), the functional CLT of
\citet{ZhouWu2009} implies~\eqref{app:eq:bernstein-cond} with an
extra $\log \nval$ factor in the exponent. Theorem~\ref{app:thm:uc}
holds with rate $O_p(\sqrt{(\log K_T \log T)/\nval}) = o_p(R_T^\star)$
under Assumption~\ref{app:ass:grid}.

\paragraph{Heavy-tailed scores.}
For polynomial tails $\Prob(\hat{S}_{t,h} > x) \asymp x^{-p}$ with
$p > 2(2\beta+1)/\beta$, the Fuk--Nagaev inequality \citep{Rio2017}
replaces Bernstein's, and Theorem~\ref{app:thm:uc} holds under
$\log K_T = o(T^{p/(2\beta+1)})$, which is weaker than
Assumption~\ref{app:ass:grid} for $p > 1$. At $\beta = 1$ the
moment condition $p > 6$ is satisfied by stationary GARCH(1,1)
models under standard parameter restrictions \citep{CarrascoChen2002}.

\end{document}